\documentclass[12pt]{article}
\pdfoutput=1
\usepackage{amsmath}
\usepackage{graphicx}
\usepackage{enumerate}
\usepackage{natbib}
\usepackage{url} 

\newcommand{\blind}{1}

\usepackage[english]{babel}
\usepackage{amssymb,amsmath,amsthm}
\usepackage{bbm}
\usepackage{dsfont}
\usepackage{xcolor}
\usepackage{verbatim}
\usepackage{subcaption}
\usepackage{algorithm}
\usepackage[noend]{algpseudocode}

\usepackage{ulem}
\usepackage[T1]{fontenc}
\usepackage{breakcites}

\usepackage{fonttable}

\usepackage{bm}
\usepackage[colorlinks]{hyperref}
\usepackage[nameinlink,capitalise]{cleveref}
\newtheorem{theorem}{Theorem}

\newtheorem{lemma}{Lemma}
\newtheorem{proposition}{Proposition}
\newtheorem{definition}{Definition}
\newcommand\independent{\protect\mathpalette{\protect\independenT}{\perp}}
\newcommand\numberthis{\addtocounter{equation}{1}\tag{\theequation}}
\def\independenT#1#2{\mathrel{\rlap{$#1#2$}\mkern2mu{#1#2}}}
\usepackage{mathtools} 

\DeclareMathOperator*{\argmaxA}{arg\,max}

\newcommand{\bA}{\textbf{A}}

\newcommand{\bD}{\textbf{D}}

\newcommand{\bs}{\textbf{s}}
\newcommand{\bS}{\textbf{S}}
\newcommand{\bH}{\textbf{H}}
\newcommand{\bM}{\textbf{M}}

\newcommand{\bQ}{\textbf{Q}}
\newcommand{\bI}{\textbf{I}}

\newcommand{\bT}{\textbf{T}}

\newcommand{\bW}{\textbf{W}}
\newcommand{\bx}{\textbf{x}}
\newcommand{\bX}{\textbf{X}}

\newcommand{\bR}{\textbf{R}}

\newcommand{\whitenS}{\mathbf S^{W}}
\newcommand{\whitens}{\mathbf s^{W}}

\newcommand{\bSigma}{ \mbox{\boldmath $\Sigma$} }

\allowdisplaybreaks
\usepackage[shortlabels]{enumitem}

\RequirePackage{lineno}

\usepackage{chngcntr}
\usepackage{apptools}
\AtAppendix{\counterwithin{lemma}{section}}

\begin{document}

\def\spacingset#1{\renewcommand{\baselinestretch}%
{#1}\small\normalsize} \spacingset{1}


\if1\blind
{
	\title{\bf Inferring independent sets of Gaussian variables after thresholding correlations}
 \author{Arkajyoti Saha$^{1}$\footnote{Corresponding email: asaha3@uw.edu, AS, DW, and JB were supported by NIH Grant R01GM123993}, Daniela Witten$^{1,2}$\footnote{DW was supported by Simons Investigator Award No. 560585, NIH Grant R01EB026908, and NIH Grant R01DA047869}, Jacob Bien$^{3}$\footnote{JB was supported by NSF CAREER Award DMS-1653017}}
 
 \date{
	$^1$ Department of Statistics, University of Washington\\
	$^2$ Department of Biostatistics, University of Washington\\
	$^3$ Department of Data Sciences and Operations, University of Southern California\\[2ex]
}
	\maketitle
} \fi

\if0\blind
{
	\bigskip
	\bigskip
	\bigskip
	\title{\bf Inferring independent sets of Gaussian variables after thresholding correlations}
	\date{}
		\maketitle
} \fi

\begin{abstract}
We consider testing whether a set of Gaussian variables, selected from the data, is independent of the remaining variables. We assume that this set is selected via a very simple approach that is commonly used across scientific disciplines: we select a set of variables for which the correlation with all variables outside the set falls below some threshold. Unlike other settings in selective inference, failure to account for the selection step leads, in this setting, to excessively conservative (as opposed to anti-conservative) results. Our proposed test properly accounts for the fact that the set of variables is selected from the data, and thus is not overly conservative. To develop our test, we condition on the event that the selection resulted in the set of variables in question. To achieve computational tractability, we develop a new characterization of the conditioning event in terms of the canonical correlation between the groups of random variables. In simulation studies and in the analysis of gene co-expression networks, we show that our approach has much higher power than a ``naive'' approach that ignores the effect of selection. 
\end{abstract}

\noindent%
\vfill

\newpage
\spacingset{1} 

\section{Introduction}

A vast statistical literature focuses on estimation of, and inference on, the sparsity structure of a covariance or correlation matrix (see, e.g.  \citet{bickel2008covariance,drton2008sinful,lam2009sparsistency,rothman2009generalized,bien2011sparse,cai2012adaptive,liu2012high,rothman2012positive,xue2012positive,zhao2012huge,liu2014sparse,serra2018robust} and references therein). In this article, we focus on a particular type of sparsity in which the matrix has a block diagonal structure (up to a permutation of the rows and columns). This type of sparsity arises in a number of scientific fields. For example, in the study of gene co-expression networks, discovery of independent ``gene sets'' plays a pivotal role (e.g. \citealt{freeman2007construction}). In genetics, discovery of haplotype blocks often reduces to partitioning the linkage disequilibrium matrix into approximately independent blocks \citep{patil2001blocks}. In neuroscience, it is of interest to discover correlated brain regions in brain functional connectivity networks \citep{bordier2017graph}.

 We consider a simple approach to estimating block diagonal structure in a correlation matrix. Our paper's primary contribution is a novel approach to testing whether two sets of variables that appear in different estimated blocks are, in fact, uncorrelated. For example, is an estimated gene set in a gene co-expression network really uncorrelated with the remaining genes? This is a challenging question because the same data is used both to estimate the gene set and to test a hypothesis based on this estimated gene set. The problem falls into the selective inference framework, which has been studied in various other contexts by \citet{fithian2014optimal,lee2014exact,loftus2015selective,lee2016exact,yang2016selective,tibshirani2016exact,suzumura2017selective,hyun2018exact,tian2018selective,hyun2018exact,jewell2019testing,chen2020valid,gao2020selective,neufeld2021tree,chen2022selective}, among others. However, unlike those other papers, in our setting, failure to account for double-use of data leads to a loss of power rather than loss of type I error control. 

The rest of this paper is organized as follows. In Section \ref{sec:prob_def}, we briefly review some results related to testing hypotheses associated with submatrices of a covariance matrix, define the notion of selective type I error in this context, and then provide a general expression for constructing p-values that are valid in this selective sense.  In Section \ref{section:characterization_of_conditioning_set}, we present a data-adaptive procedure for estimating diagonal blocks in a correlation matrix, and consider testing the null hypothesis of no correlation between an estimated block and the remaining variables. Furthermore, we show that the event corresponding to the selection of a particular hypothesis can be characterized analytically. In Section \ref{sec:summary}, we apply this characterization to develop a selective inference approach to test the null hypothesis of interest. Simulation results in Section \ref{sec:simulation} and an analysis of gene co-expression networks in Section \ref{sec:gene_coexpression} demonstrate the value of the proposed approach. We close with a discussion in Section \ref{sec:discussion}. 

\section{Problem definition}
\label{sec:prob_def}

\subsection{Hypothesis testing for pre-specified groups of variables}
\label{sec:prob_def_hyp}

For a positive definite $p \times p$ matrix $\bSigma$, consider a data matrix consisting of $n > p$ independent observations of a $p$-dimensional multivariate Gaussian,

\begin{equation}
\label{eqn:distribution_data}
    \bX \sim \mathcal{MN}_{n \times p}\left(\mathbf 0, \mathbf I_n, \mathbf \bSigma_{p \times p}\right).
\end{equation}
For a predetermined subset $\mathcal P \subset \left\{1,2, \ldots, p \right\}$ of variables, consider testing for dependence between the variables in $\mathcal{P}$ and the variables in its complement ${\mathcal{P}}^c := \left\{1,2, \ldots, p \right\} \setminus \mathcal P$, i.e.
\begin{equation}
\label{eqn:base_hyp}
    H_0^{\mathcal{P}} : \bSigma_{{\mathcal{P}}, {\mathcal{P}}^c} = \mathbf 0 \:\text{ versus }\: H_1^{\mathcal{P}} : \bSigma_{{\mathcal{P}}, {\mathcal{P}}^c} \neq \mathbf 0,
\end{equation}
where $\bSigma_{\mathcal{P}_1, \mathcal{P}_2}$ denotes the submatrix of $\bSigma$ with rows in $\mathcal P_1$ and columns in $\mathcal P_2$. Given a realization $\bx \in \mathbb R^{n \times p}$ of $\bX$, the classical likelihood ratio test (LRT) of \eqref{eqn:base_hyp} yields the p-value

\begin{equation}
\label{eqn:lrt}
    p_{\mathrm{LRT}}\left(\bx;  \mathcal P \right):= \mathbb{P}_{H_0^{\mathcal{P}}} \left(\frac{\left|\bS\left(\bX\right)\right|}{\left|\bS_{{\mathcal{P}}, {\mathcal{P}}}\left(\bX\right)\right|\left|\bS_{{\mathcal{P}}^c, {\mathcal{P}}^c}\left(\bX\right)\right|} \leqslant \frac{\left|\bS\left(\bx\right)\right|}{\left|\bS_{{\mathcal{P}}, {\mathcal{P}}}\left(\bx\right)\right|\left|\bS_{{\mathcal{P}}^c, {\mathcal{P}}^c}\left(\bx\right)\right|} \right),
\end{equation}
where $\bS\left(\bx\right) = \frac{1}{n}\left(\bx - \mathbf 1_n\bar{x}^\top\right)^\top\left(\bx - \mathbf 1_n\bar{x}^\top\right)$ is the sample covariance matrix of $\bx$, $\bar{x} \in \mathbb R^p$ denotes the sample mean of $\bx$, and $\left|\bM\right|$ denotes the determinant of the matrix $\bM$ (see (3) of Section 9.8 in \citealt{anderson1962introduction}). Here, $\bS_{\mathcal P_1, \mathcal P_2}$ represents the submatrix of $\bS$ with rows in $\mathcal P_1$ and columns in $\mathcal P_2$. The compact singular value decomposition (SVD) of the cross-covariance of decorrelated (``whitened'') versions of $\bx_\mathcal{P}$ and $\bx_{{\mathcal{P}}^c}$ (defined as the submatrices of $\bx$ with columns in $\mathcal P$ and ${\mathcal{P}}^c$, respectively) can be written as
\begin{equation}
    \label{eqn:CCA}
    \whitenS_{{\mathcal{P}}, {\mathcal{P}}^c} \left(\bx\right) := \left(\bS_{{\mathcal{P}}, {\mathcal{P}}}\left(\bx\right)\right)^{-\frac{1}{2}}\bS_{{\mathcal{P}}, {\mathcal{P}}^c}\left(\bx\right) \left(\bS_{{\mathcal{P}}^c, {\mathcal{P}}^c}\left(\bx\right)\right)^{-\frac{1}{2}}\overset{SVD}{=} \hat{\mathbf{A}}^{{\mathcal{P}}}\left(\bx\right) \: \hat{\bm\Lambda}^{{\mathcal{P}}}\left(\bx\right)\:  \left({\hat{\mathbf{\Gamma}}^{{{\mathcal{P}}}}}\left(\bx\right)\right)^\top,
\end{equation}
where $\hat{\bm\Lambda}^{{\mathcal{P}}}\left(\bx\right) := \mathrm{diag}\left( \hat\lambda_1\left(\bx\right), \hat\lambda_2\left(\bx\right), \ldots, \hat\lambda_{r\left( \mathcal P\right)}\left(\bx\right)\right)$ is a square diagonal matrix of\\ $r\left(\mathcal P\right) := \min\left\{\mathrm{card}\left({\mathcal{P}}\right), \mathrm{card}\left({\mathcal{P}}^c\right) \right\}$ non-negative singular values in  non-increasing order ($\hat\lambda_1\left(\bx\right) \geqslant \hat\lambda_2\left(\bx\right) \geqslant \ldots \geqslant \hat\lambda_{r\left( \mathcal P\right)}\left(\bx\right) \geqslant 0 $) and $\hat{\mathbf{A}}^{{\mathcal{P}}} \left(\bx\right)$ and $\hat{\bm{\Gamma}}^{{\mathcal{P}}}\left(\bx\right)$ are the $\mathrm{card}\left({\mathcal{P}}\right) \times r\left(\mathcal P\right)$ and $\mathrm{card}\left({\mathcal{P}}^c\right) \times r\left(\mathcal P\right) $ matrices of left and right singular vectors, respectively. 
The diagonal elements of $\hat{\bm{\Lambda}}^{{\mathcal{P}}}\left(\bx\right)$ can be interpreted as the canonical correlations between $\bx_{{\mathcal{P}}}$ and $\bx_{{\mathcal{P}}^c}$, and the columns of $\left(\bS_{{\mathcal{P}}, {\mathcal{P}}}\left(\bx\right)\right)^{-\frac{1}{2}}\hat{\mathbf{A}}^{{\mathcal{P}}}\left(\bx\right)$ and $\left(\bS_{{\mathcal{P}^c}, {\mathcal{P}^c}}\left(\bx\right)\right)^{-\frac{1}{2}}\hat{\mathbf{\Gamma}}^{{{\mathcal{P}}}}\left(\bx\right)$ are the left and right canonical vectors. 

Using the notation in \eqref{eqn:CCA}, the LRT statistic in \eqref{eqn:lrt} can be rewritten as

\begin{equation}
\label{eqn:test_stat_simplified}
    \frac{\left|\bS\left(\bx\right)\right|}{\left|\bS_{{\mathcal{P}}, {\mathcal{P}}}\left(\bx\right)\right|\left|\bS_{{\mathcal{P}}^c, {\mathcal{P}}^c}\left(\bx\right)\right|} = \left|\bI_{r\left(\mathcal P\right)} - \left(\hat{\bm{\Lambda}}^{{\mathcal{P}}}\left(\bx\right)\right)^2\right|,
\end{equation}
and the associated p-value can be rewritten as
\begin{equation}
\label{eqn:base_hyp_simplified}
    p_{\mathrm{LRT}}\left(\bx;  \mathcal P \right) = \mathbb{P}_{H_0^{\mathcal{P}}} \left(\left|\bI_{r\left(\mathcal P\right)} - \left(\hat{\bm{\Lambda}}^{{\mathcal{P}}}\left(\bX\right)\right)^2\right| \leqslant \left|\bI_{r\left(\mathcal P\right)} - \left(\hat{\bm{\Lambda}}^{{\mathcal{P}}}\left(\bx\right)\right)^2\right|\right).
\end{equation}
Details are in Supplementary Material Section \ref{Appendix:CCA}. Under $H_0^\mathcal{P}$ and the assumption that $n > p$, the test statistic follows a Wilks' lambda distribution (see page 288 of \citealt{kent1979multivariate}). We will later use Proposition \ref{prop:density_singular_values} to sample from the Wilks' lambda distribution to evaluate $p_{\mathrm{LRT}}\left(\bx;  \mathcal P \right)$ in \eqref{eqn:base_hyp_simplified}. 

\begin{proposition}
\label{prop:density_singular_values}
Under $H_0^\mathcal{P}$ and the assumption that $n > p$, the following statements hold regarding the sample canonical correlations $\left(\hat\lambda_1, \hat\lambda_2, \ldots, \hat\lambda_{r\left(\mathcal P\right)} \right)$ defined in \eqref{eqn:CCA}: 
\begin{enumerate}[(i)]
    \item The sample canonical correlations have joint density, 
\begin{equation}
\label{eqn:density_singular_values}
   \begin{aligned}
    f\left(\hat\lambda_1, \hat\lambda_2, \ldots, \hat\lambda_{r\left(\mathcal P\right)}\right) =& \pi^{\frac{1}{2}\left(r\left(\mathcal P\right)\right)^2}\frac{2^{r\left(\mathcal P\right)}\Gamma_{r\left(\mathcal P\right)}\left(\frac{n}{2}\right)}{\Gamma_{r\left(\mathcal P\right)}\left(\frac{n-p + r\left(\mathcal P\right)}{2}\right)\Gamma_{p}\left(\frac{1}{2}r\left(\mathcal P\right)\right)\Gamma _{p}\left(\frac{1}{2}\left(p- r\left(\mathcal P\right)\right)\right)} \times\\ 
    &\prod_{i = 1}^{r\left(\mathcal P\right)}\left\{\hat\lambda_i^{\left(p - 2r\left(\mathcal P\right)\right)}\left(1 - \hat\lambda_i^2\right)^{\frac{1}{2}\left(n - p - 2\right)} \right\}\prod_{i < j}\left(\hat\lambda_i^2 - \hat\lambda_j^2\right), 
   \end{aligned}
\end{equation}
supported on $\left\{1 \geqslant \hat\lambda_1 \geqslant \hat\lambda_2 \geqslant \ldots \geqslant \hat\lambda_{r\left(\mathcal P\right)} \geqslant 0 \right\}$.

\item The sample canonical correlations have the same distribution as \\ $\left( \sqrt\frac{\Psi_1}{1 + \Psi_1}, \sqrt\frac{\Psi_2}{1 + \Psi_2}, \ldots, \sqrt\frac{\Psi_{r\left(\mathcal P\right)}}{1 + \Psi_{r\left(\mathcal P\right)}}\right)$, where $\left(\Psi_1, \Psi_2, \ldots, \Psi_{r\left(\mathcal P\right)} \right)$ are the eigenvalues of $\mathbf W \mathbf T^{-1}$, for $\mathbf W \sim \mathrm{Wishart}_{r\left(\mathcal{{P}}\right)}\left(\mathbf I_{r\left(\mathcal{{P}}\right)}, p - r\left(\mathcal{{P}}\right)\right)$,  $\mathbf T \sim \mathrm{Wishart}_{r\left(\mathcal{{P}}\right)}\left(\mathbf I_{r\left(\mathcal{{P}}\right)}, n - 1 - p + r\left(\mathcal{{P}}\right)\right)$, and $\mathbf W  \independent \mathbf T$. 

\item If $r\left({\mathcal{P}}\right) = 1$, then the sample canonical correlation $\hat\lambda_1$ is the sample multiple correlation coefficient for the regression of the single variable onto the remaining $p - 1$ variables, and

\begin{equation}
    \hat\lambda_1 ^2 \sim \mathrm{Beta}\left( \frac{p - 1}{2}, \frac{n - p}{2}\right).
\end{equation}

\end{enumerate}

\end{proposition}

Here,  (i) gives us the joint density of $\left(\hat\lambda_1, \hat\lambda_2, \ldots, \hat\lambda_{r\left(\mathcal P\right)} \right)$. In Section \ref{sec:summary}, we will use (ii) to sample from the joint distribution of $\left(\hat\lambda_1, \hat\lambda_2, \ldots, \hat\lambda_{r\left(\mathcal P\right)} \right)$, as the joint density in (i) is not amenable to sampling. We will make use of (iii) to bypass sampling in the special case of $r\left({\mathcal{P}}\right) = 1$. We defer the proof of Proposition \ref{prop:density_singular_values} to Supplementary Material Section \ref{Appendix:density}. Crucially, under $H_0^\mathcal{P}$, the distribution of the sample canonical correlations depends only on $n$, $p$, and $r\left({\mathcal{P}}\right)$. 

Next, we explore how the situation changes if we consider a data-dependent subset of variables, instead of a predetermined one as in the case of \eqref{eqn:base_hyp}.

\subsection{What changes when the groups of variables are chosen from the data?}

We now suppose that $\mathcal{P}$ is a function of the data, and specifically of $\bS\left(\bx\right)$. That is, we define $\mathcal P : \mathcal S_{++}^p \mapsto \mathcal{B}_{p}$, where $\mathcal S_{++}^p$ denotes the set of all symmetric positive definite matrices, and $\mathcal{B}_{p}$ denotes the set of all possible partitions of $\left\{1,2,\ldots, p \right\}$. For $\hat{\mathcal{P}} \in \mathcal{P}\left(\bS\left(\bx\right)\right)$, we wish to test 

\begin{equation}
\label{eqn:H0_selection}
    H_0^{\hat{\mathcal{P}}} : \bSigma_{\hat{\mathcal{P}}, \hat{\mathcal{P}}^c} = \mathbf 0 \:\:\:\:\text{versus}\:\:\: H_1^{\hat{\mathcal{P}}} : \bSigma_{\hat{\mathcal{P}}, \hat{\mathcal{P}}^c} \neq \mathbf 0.
\end{equation}
We assume without loss of generality that $\hat{\mathcal{P}} = \left\{ 1,2,\ldots, \mathrm{card}\left({\hat{\mathcal{P}}}\right)\right\}$ and\\ $\hat{\mathcal{P}}^c = \left\{ \mathrm{card}\left({\hat{\mathcal{P}}}\right) + 1, \ldots, p \right\}$. Since $\hat{\mathcal{P}} \in \mathcal{P}\left(\bS\left(\bx\right)\right)$, the null hypothesis in \eqref{eqn:H0_selection} is a function of the data. The classical test statistic in \eqref{eqn:lrt} does not follow a Wilks' lambda distribution under the null hypothesis \eqref{eqn:H0_selection}. Ignoring this fact can lead to a tremendous loss of power. To demonstrate this in an example, we simulate data with $p = 6$ and $n = 9$ for a completely dense $\bSigma$, so that $H_1^{\hat{\mathcal{P}}} : \bSigma_{\hat{\mathcal{P}}, \hat{\mathcal{P}}^c} \neq \mathbf 0$ holds for all $\hat{\mathcal P}$. Figures \ref{Fig:power}(a)-(b) display the heatmaps of the absolute values of the entries of the population correlation matrix and the sample correlation matrix. We tested \eqref{eqn:H0_selection} for $\hat{\mathcal{P}}$ obtained via the thresholding procedure described in Algorithm \ref{algo:Selection_procedure}; this yielded $\hat{\mathcal{P}} = \left\{1,2, 3 \right\}$. Ignoring that $\hat{\mathcal{P}}$ was selected based on the data and using Wilks' lambda as the reference distribution for the likelihood ratio test statistic yielded a p-value $\left(p_{\mathrm{LRT}}\right)$ of $0.773$, whereas our proposed selective inference approach produced a p-value of $0.015$.  Figure \ref{Fig:power}(c) shows the distribution of p-values over $10,000$ data realizations using the two approaches. We see that our selective p-values tend to be small --- as expected, since the null hypothesis does not hold --- whereas the classical p-values are stochastically greater than a uniform.

\begin{figure}[t!]
    \centering
    \begin{subfigure}[t]{0.5\textwidth}
    \label{Fig:cost_known_sinpi}
        \centering
        \includegraphics[height=1.75in]{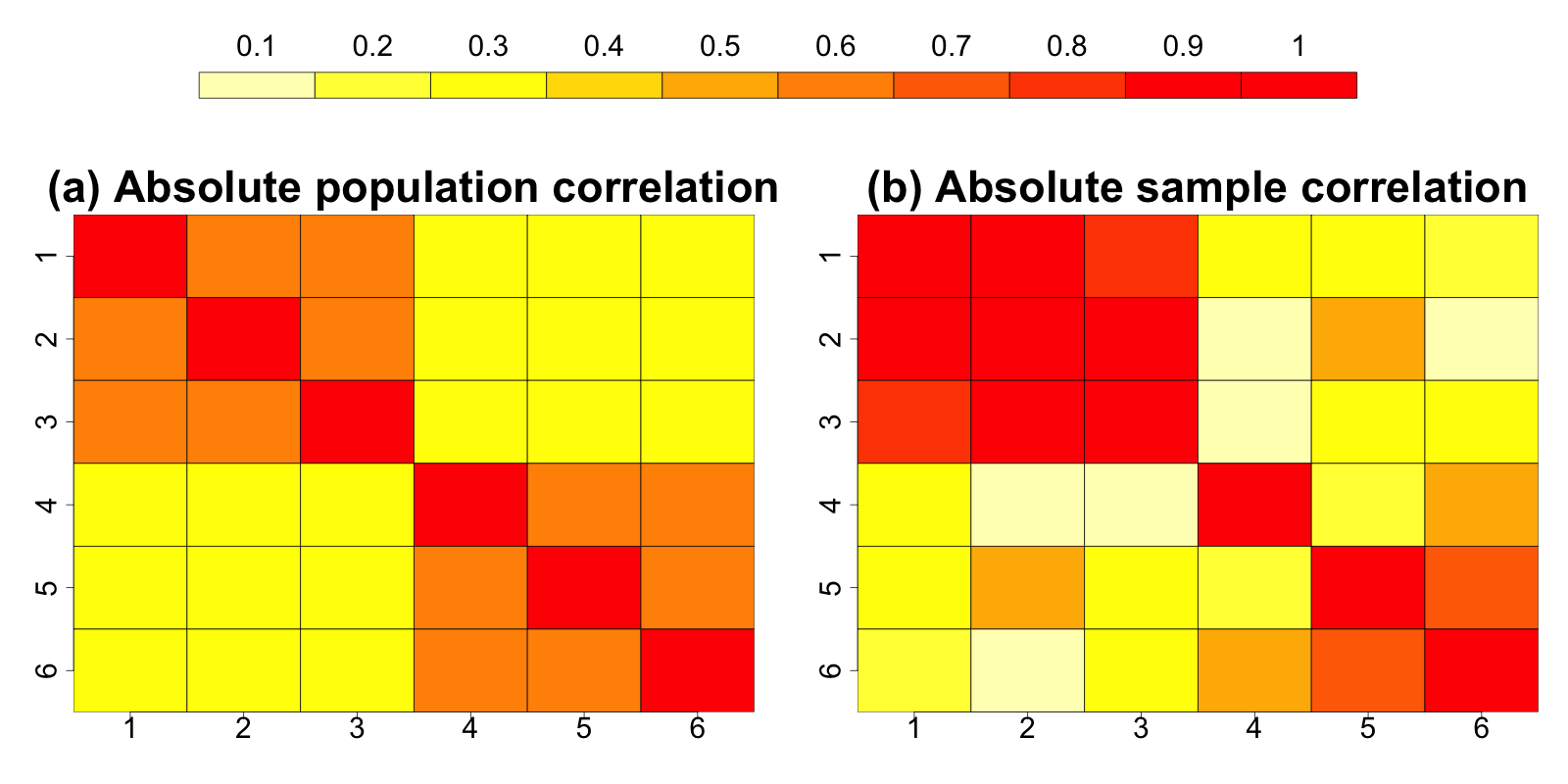}
    \end{subfigure}%
    \hskip -0.001cm 
    \begin{subfigure}[t]{0.5\textwidth}
    \label{Fig:}
        \centering
        \includegraphics[height=1.75in]{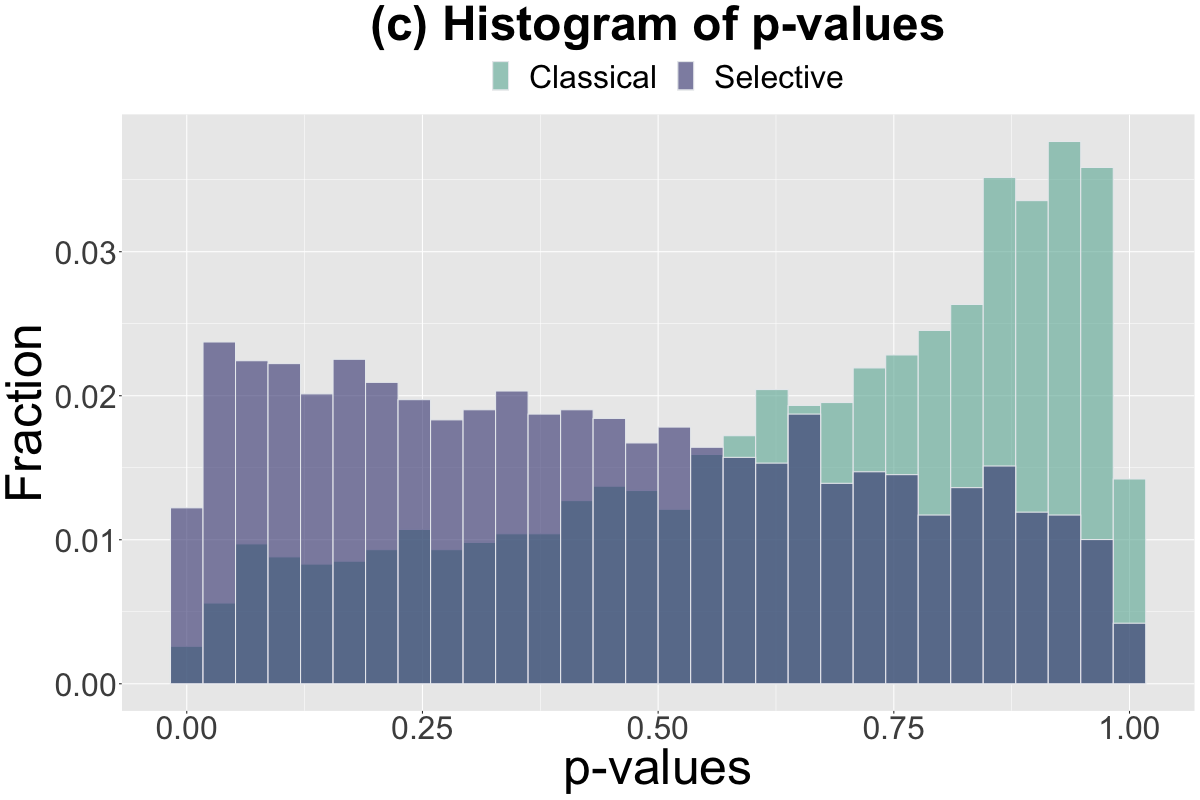}
    \end{subfigure}
    \caption{Heatmap of the absolute values of the entries of the (a) population correlation matrix and (b) sample correlation matrix, for $p = 6$ and $n = 9$. Histograms of the classical LRT p-values and the selective LRT p-values are in (c). }\label{Fig:power}
\end{figure}

At first glance, the fact that the classical test has low power in this scenario may seem counter-intuitive, since failure to account for selection typically leads to an \textit{inflated} type I error rate, as opposed to low power, in related settings \citep{fithian2014optimal}. However, in this setting, the procedure $\mathcal P\left(\cdot\right)$ that selects $\hat{\mathcal P}$ (as described in Algorithm \ref{algo:Selection_procedure}) does so because of its low correlation with $\hat{\mathcal P}^c$. Thus, the likelihood ratio test statistic will be stochastically greater than the Wilks' lambda distribution under $H_0^{\hat{\mathcal{P}}} : \bSigma_{\hat{\mathcal{P}}, \hat{\mathcal{P}}^c} = \mathbf 0$, yielding an overly conservative p-value. 

\subsection{Selective inference for the covariance matrix}
\label{section:Matrix}

\citet{fithian2014optimal} argues that when a null hypothesis is \textit{selected} from the data, we should control the probability of a false rejection conditional on \textit{having selected this null hypothesis}, i.e. the \textit{selective type I error rate}. In the context of \eqref{eqn:H0_selection}, the selective type I error rate is defined as follows:

\begin{definition}
\textbf{(Selective type I error)} A test of $H_0^{\mathcal R} : \bSigma_{\mathcal R, \mathcal R^c} = \mathbf 0$ controls the selective type I error if, for every subset ${\mathcal{R}}$ of $\left\{1,2,\ldots, p \right\}$,
\begin{equation}
\label{eqn:selective_type_1}
    \mathbb P_{ H_0^{{\mathcal{R}}}} \left( \text{reject } H_0^{{\mathcal{R}}} \text{ at level } \alpha \bigg | {\mathcal{R}} \in \mathcal{P}\left(\mathbf S\left(\mathbf X\right)\right) \right) \leqslant \alpha; \:\forall \:0 \leqslant \alpha \leqslant 1.
\end{equation}
\end{definition}

By the probability integral transform, it can be shown that a conditional version of \eqref{eqn:base_hyp_simplified},

\begin{equation}
\label{eqn:selective_pvalue_initial}
     p_0\left(\bx; \hat{\mathcal{P}}\right) := \mathbb{P}_{ H_0^{\hat{\mathcal{P}}}} \left(\left|\bI - \left(\hat{\bm{\Lambda}}^{\hat{\mathcal{P}}}\left(\bX\right)\right)^{2}\right| \leqslant \left|\bI - \left(\hat{\bm{\Lambda}}^{\hat{\mathcal{P}}}\left(\bx\right)\right)^{2}\right| \bigg | \hat{\mathcal{P}} \in \mathcal{P} \left(\bS\left(\bX\right)\right)\right),
\end{equation}
controls the selective type I error in \eqref{eqn:selective_type_1}. However, \eqref{eqn:selective_pvalue_initial} is not computationally tractable, because the distribution of $\left|\bI - \left(\hat{\bm{\Lambda}}^{\hat{\mathcal{P}}}\left(\bX\right)\right)^{2}\right|$, conditional on $\hat{\mathcal{P}} \in \mathcal{P} \left(\bS\left(\bX\right)\right)$, depends on parameters that are unknown even under the null hypothesis. To bypass this problem, we employ a common approach in the selective inference literature (see e.g. \citealt{lee2016exact}), in which we condition on some additional information. In particular, we propose a p-value for $H_0^{\hat{\mathcal{P}}}$ as
\begin{equation}
\label{eqn:selective_pvalue_definition}
\begin{aligned}
        p\left(\bx; \hat{\mathcal{P}}\right) := \mathbb{P}_{H_0^{\hat{\mathcal{P}}}} \bigg(\left|\bI - \left(\hat{\bm{\Lambda}}^{\hat{\mathcal{P}}}\left(\bX\right)\right)^2\right| \leqslant \left|\bI - \left(\hat{\bm{\Lambda}}^{\hat{\mathcal{P}}}\left(\bx\right)\right)^2\right| \bigg |& \hat{\mathcal{P}} \in \mathcal{P}\left(\bS\left(\bX\right)\right),\\ \bS_{\hat{\mathcal{P}}, \hat{\mathcal{P}}} \left(\bX\right) = \bS_{\hat{\mathcal{P}}, \hat{\mathcal{P}}} \left(\bx\right), \bS_{\hat{\mathcal{P}}^c, \hat{\mathcal{P}}^c} \left(\bX\right) = \bS_{\hat{\mathcal{P}}^c, \hat{\mathcal{P}}^c} \left(\bx\right), &\:\hat{\mathbf{A}}^{\hat{\mathcal{P}}}\left(\bX\right) = \hat{\mathbf{A}}^{\hat{\mathcal{P}}}\left(\bx\right), \hat{\bm{\Gamma}}^{\hat{\mathcal{P}}}\left(\bX\right) = \hat{\bm{\Gamma}}^{\hat{\mathcal{P}}}\left(\bx\right)\bigg),
\end{aligned}
\end{equation}
where we recall the SVD in \eqref{eqn:CCA}. 
The intuition behind this conditioning set is as follows:

\begin{enumerate}
    \item \textbf{Fixed within-group covariance: } Due to the nature of the null hypothesis in \eqref{eqn:H0_selection}, we are primarily interested in inter-group correlation. Hence we keep the intra-group covariance fixed, i.e., we only consider $\bX$ for which $\bS_{\hat{\mathcal{P}}, \hat{\mathcal{P}}}\left(\bX\right) = \bS_{\hat{\mathcal{P}}, \hat{\mathcal{P}}}\left(\bx\right)$ and $\bS_{\hat{\mathcal{P}}^c, \hat{\mathcal{P}}^c}\left(\bX\right) = \bS_{\hat{\mathcal{P}}^c, \hat{\mathcal{P}}^c}\left(\bx\right)$. 
    
    \item \textbf{Fixed between-group distance direction: } Inspecting \eqref{eqn:CCA} and \eqref{eqn:test_stat_simplified}, we see that the only part of $\whitenS_{\hat{\mathcal{P}}, \hat{\mathcal{P}}^c}$ that is involved in the test statistic is $\hat{\bm\Lambda}^{\hat{\mathcal{P}}}$. The left and right singular vectors of  $\whitenS_{\hat{\mathcal{P}}, \hat{\mathcal{P}}^c}$ in \eqref{eqn:CCA}, i.e. $\hat{\mathbf{A}}^{\hat{\mathcal{P}}}$ and $\hat{\bm{\Gamma}}^{\hat{\mathcal{P}}}$, are not involved in the test statistic. Hence, in \eqref{eqn:selective_pvalue_definition}, we only consider $\bX$ for which $\hat{\mathbf{A}}^{\hat{\mathcal{P}}}\left(\bX\right) = \hat{\mathbf{A}}^{\hat{\mathcal{P}}}\left(\bx\right)$ and $ \hat{\bm{\Gamma}}^{\hat{\mathcal{P}}}\left(\bX\right) = \hat{\bm{\Gamma}}^{\hat{\mathcal{P}}}\left(\bx\right)$.
\end{enumerate}

The next theorem establishes that the conditioning set in \eqref{eqn:selective_pvalue_definition} can be  written in terms of constraints on the sample canonical correlations, and that $p\left(\mathbf x; \hat{\mathcal{P}} \right)$  controls the selective type I error rate.
\begin{theorem}
\label{thm:main}
For any realization $\mathbf x$ of $\mathbf X$ in \eqref{eqn:distribution_data} and for an arbitrary subset $\mathcal{R}$ of $\left\{1,2,\ldots, p \right\}$, we can express $p\left(\mathbf x; {\mathcal{R}} \right)$ defined in \eqref{eqn:selective_pvalue_definition} as 
\begin{equation}
\label{eqn:thm_main}
    p\left(\mathbf x; {\mathcal{R}} \right) = \mathbb{P}_{ H_0^{\mathcal{R}} }\left(\prod_{i = 1}^{r\left({\mathcal{R}}\right)} \left(1 - \lambda_i^2\right) \leqslant  \prod_{i = 1}^{r\left({\mathcal{R}}\right)} \left(1 - \left(\hat\lambda_i ^\mathcal R\left(\mathbf x\right)\right)^2\right) \bigg | \left(\lambda_1, \lambda_2, \ldots, \lambda_{r\left({\mathcal{R}}\right)}\right) \in {\mathcal{G}}\left(\mathbf x; {\mathcal{R}} \right)\right),
\end{equation}
where $\left\{ \lambda_i\right\}_{i=1}^{r\left({\mathcal{R}}\right)}$ has joint density $f\left(\lambda_1, \lambda_2, \ldots, \lambda_{r\left({\mathcal{R}}\right)}\right)$ defined in \eqref{eqn:density_singular_values},  $\left\{ \hat\lambda_i^\mathcal R\left(\mathbf x\right)\right\}_{i=1}^{r\left({\mathcal{R}}\right)}$ are the diagonal entries of $\hat{\bm{\Lambda}}^{\mathcal R} \left(\mathbf x\right)$, and 
\small
\begin{equation}
    \label{eqn:conditioning_set}
   {\mathcal{G}}\left(\mathbf x; {\mathcal{R}} \right):= \left\{ \left(\lambda_1, \lambda_2, \ldots, \lambda_{r\left({\mathcal{R}}\right)} \right) \in [0, 1]^{r\left({\mathcal{R}}\right)}: \lambda_1 \geqslant \lambda_2 \geqslant \ldots \geqslant \lambda_{r\left({\mathcal{R}}\right)} ; {\mathcal{R}} \in \mathcal P\left(\mathbf S'\left(\mathbf x, {{\mathcal{R}}}, \left\{\lambda_i\right\}_{i= 1}^{r\left({\mathcal{R}}\right)}\right)\right)\right\},
\end{equation}
\normalsize
where $\mathbf S'\left(\mathbf x, {\mathcal{R}}, \left\{\lambda_i\right\}_{i= 1}^{r\left({\mathcal{R}}\right)}\right)$ is a modified version of the sample covariance matrix with a perturbed off-diagonal block,
\begin{equation}
\label{eqn:new_S}
\mathbf S_{{\mathcal{R}}, {\mathcal{R}}^c}'\left(\mathbf x, {\mathcal{R}}, \left\{\lambda_i\right\}_{i= 1}^{r\left({\mathcal{R}}\right)}\right) := \left(\mathbf S_{{\mathcal{R}}, {\mathcal{R}}}\left(\mathbf x\right)\right)^{\frac{1}{2}} \hat{\mathbf{A}}^{{\mathcal{R}}} \left(\mathbf x\right) \mathrm{diag}\left(\lambda_1, \lambda_2, \ldots, \lambda_{r\left({\mathcal{R}}\right)}\right) \left(\hat{\mathbf{\Gamma}}^{{\mathcal{R}}}\left(\mathbf x\right)\right) ^\top \left(\mathbf S_{{\mathcal{R}}^c, {\mathcal{R}}^c}\left(\mathbf x\right)\right)^{\frac{1}{2}}.
\end{equation}
Furthermore, rejecting $H_0^{{\mathcal{R}}}$ if $ p\left(\mathbf x; {\mathcal{R}}\right) \leqslant \alpha$  controls the selective type I error rate at level $\alpha$, i.e. 
\begin{equation}
\label{eqn:thm_type_1}
    \mathbb P_{ H_0^{{\mathcal{R}}}} \left(p\left(\mathbf X; {\mathcal{R}}\right) \leqslant \alpha \bigg | {\mathcal{R}} \in \mathcal{P}\left(\mathbf S\left(\mathbf X\right)\right) \right) \leqslant \alpha;\: \forall \:0 \leqslant \alpha \leqslant 1.
\end{equation}
\end{theorem}
\noindent Recall that $\hat{\mathbf{A}}^{{\mathcal{R}}} \left(\mathbf x\right), \hat{\bm\Lambda}^{{\mathcal{R}}} \left(\mathbf x\right)$ and $\hat{\bm\Gamma}^{{\mathcal{R}}} \left(\mathbf x\right)$ were defined in \eqref{eqn:CCA}. We defer the proof of Theorem \ref{thm:main} to  Supplementary Material Section \ref{Appendix:Main_theorem}. Theorem \ref{thm:main} demonstrates that computation of the p-value in \eqref{eqn:selective_pvalue_definition} boils down to characterizing $ {\mathcal{G}}\left(\mathbf x; \hat{\mathcal{P}} \right)$ in \eqref{eqn:conditioning_set}.

Furthermore, ${\mathcal{G}}\left(\mathbf x; \hat{\mathcal{P}} \right)$ in \eqref{eqn:conditioning_set} describes the points $\left\{\lambda_i\right\}_{i= 1}^{r\left(\hat{\mathcal{P}}\right)}$ in the $r\left(\hat{\mathcal{P}}\right)$-dimensional unit hyper-cube for which the coordinates are in non-increasing order, and for which applying a variable grouping method $\mathcal P \left(\cdot\right)$ to the perturbed covariance matrix $\mathbf S'\left(\mathbf x, \hat{\mathcal{P}}, \left\{\lambda_i\right\}_{i= 1}^{r\left(\hat{\mathcal{P}}\right)}\right)$ yields $\hat{\mathcal{P}}$.  Characterization of this set depends on the procedure used for obtaining the block diagonal structure of the correlation matrix, i.e. on $\mathcal P \left(\cdot\right)$. In the next section, we focus on a specific data-adaptive procedure for $\mathcal P \left(\cdot\right)$, and on the corresponding characterization of ${\mathcal{G}}\left(\mathbf x; \hat{\mathcal{P}} \right)$ in \eqref{eqn:conditioning_set}. 

\section{Characterization of the conditioning set}
\label{section:characterization_of_conditioning_set}

\subsection{Procedure for obtaining groups of uncorrelated variables}
\label{subsection:thresholding_procedure}
We present a simple procedure for identifying groups of uncorrelated variables, i.e., for discovering block diagonal structure of the covariance matrix. We assume that there are no ties between the off-diagonal entries of the sample correlation matrix, which holds with probability $1$. We denote the correlation matrix corresponding to the sample covariance matrix $\bS = \bS\left(\bx\right)$ as $\mathbf{R}\left(\mathbf S\right)$, with $\left(i,j\right)^{th}$ element $\bR_{i,j}\left(\bS\right) = \frac{\bS_{i,j}}{\sqrt{\bS_{i,i} \bS_{j,j}}}$, for $ i,j \in \left\{1,2,\ldots, p\right\}$. Let $\mathds{1}\left\{ \mathcal A\right\}$ denote the indicator function of the event $\mathcal A$. Algorithm \ref{algo:Selection_procedure} summarizes the procedure for obtaining groups of uncorrelated variables. Note that if the variables have an intrinsic ordering, then in line 5 of Algorithm \ref{algo:Selection_procedure} we discard groups of variables that violate that ordering.

\begin{algorithm}[h]
	\caption{Procedure for obtaining groups of uncorrelated variables}\label{algo:Selection_procedure}
	 \textbf{Input:} Sample covariance matrix $\bS \left(\bx\right)$; threshold $c$.
	\begin{algorithmic}[1]
		\Procedure{}{}
		\State Create the adjacency matrix $\mathbf D_{i,j} = \mathds{1}\left\{ \left|\bR_{i,j}\left(\bS\left(\bx \right)\right)\right| > c\right\}$;
		\State Identify the connected components $\mathcal{C}_1, \mathcal{C}_2, \ldots, \mathcal{C}_K$ of the graph corresponding to the \newline \hspace*{1.65em} adjacency matrix;
		\If{\text{variables are ordered}}
		\State $\mathcal{P}\left(\bS\left(\bx\right)\right)  = \left\{\mathcal{C}_k:  \mathcal{C}_k \text{ consists exclusively of consecutive variables}\right\}$;
		\Else
		\State $\mathcal{P}\left(\bS\left(\bx\right)\right)  = \left\{\mathcal{C}_1, \mathcal{C}_2, \ldots, \mathcal{C}_K \right\}$.
		\EndIf
		\EndProcedure
	\end{algorithmic}
\end{algorithm}

\subsection{Calculation of the conditioning set}
In this subsection, we characterize the conditioning set ${\mathcal{G}}\left(\mathbf x; \hat{\mathcal{P}} \right)$ in \eqref{eqn:conditioning_set} for the function $\mathcal P \left(\cdot\right)$ defined in Algorithm \ref{algo:Selection_procedure}. 

\begin{proposition}
\label{prop:reformulate_conditioning_set}
For $\hat{\mathcal P}\in \mathcal{P} \left(\mathbf S\left(\mathbf x\right)\right) $, we have that
\begin{equation}
\label{eqn:polytope_proof}
\begin{aligned}
 {\mathcal{G}}\left(\mathbf x; \hat{\mathcal{P}} \right) = \bigg\{ \left(\lambda_1, \lambda_2, \ldots, \lambda_{r\left(\hat{\mathcal{P}}\right)}\right) \in [0, 1]^{r\left(\hat{\mathcal{P}}\right)} :  &\:\lambda_1 \geqslant \lambda_2 \geqslant \ldots \geqslant \lambda_{r\left(\hat{\mathcal{P}}\right)},\\  & c \geqslant \max_{i' \in \hat{\mathcal P}, j' \in \hat{\mathcal P}^c}\left|\mathbf R_{i',j'}\left(\mathbf S'\left(\mathbf x, \hat{\mathcal{P}}, \left\{\lambda_i\right\}_{i= 1}^{r\left(\hat{\mathcal{P}}\right)}\right)\right)\right|\bigg\}.
\end{aligned}
\end{equation}
\end{proposition}

\noindent The proof is in Supplementary Material Section \ref{Appendix:conditioning_set}.
The next proposition shows that the inter-group correlations appearing in the above expression are in fact linear combinations of $\lambda_1, \lambda_2, \ldots, \lambda_{r\left(\hat{\mathcal{P}}\right)}$.

\begin{proposition}
\label{prop:linear1}
The $\mathrm{card}\left(\hat{\mathcal{P}}\right) \times \left(p - \mathrm{card}\left(\hat{\mathcal{P}}\right) \right)$ submatrix $\mathbf{R}_{\hat{\mathcal{P}}, \hat{\mathcal{P}}^c}\left(\mathbf S'\left(\mathbf x, \hat{\mathcal{P}}, \left\{\lambda_i\right\}_{i= 1}^{r\left(\hat{\mathcal{P}}\right)}\right)\right)$ is linear in $\left(\lambda_1, \lambda_2, \ldots, \lambda_{r\left(\hat{\mathcal{P}}\right)}\right)$. 
\end{proposition}
\noindent The proof of this proposition is provided in Supplementary Material Section \ref{Appendix:prop_linear1}. In the next theorem, we show that $   {\mathcal{G}}\left(\mathbf x; \hat{\mathcal{P}} \right) $ in \eqref{eqn:conditioning_set} can be rewritten as a closed convex polytope in $\mathbb R^{r\left(\hat{\mathcal{P}}\right)}$.
\begin{theorem}
\label{thm:character}
For $\hat{\mathcal P} \in \mathcal{P}\left( \mathbf S \left( \mathbf x\right)\right)$ corresponding to the threshold $c$ in Algorithm \ref{algo:Selection_procedure}, there exist $ \mathbf{L} = \mathbf{L} \left(\mathbf S\left(\mathbf x\right), \hat{\mathcal P}\right) \in \mathbb R^{\left(2 \mathrm{card}\left(\hat{\mathcal{P}}\right)\left(p - \mathrm{card}\left(\hat{\mathcal{P}}\right) \right) + r\left(\hat{\mathcal{P}}\right) 
+ 1\right) \times r\left(\hat{\mathcal{P}}\right)}$ and $\mathbf g = \mathbf g\left(\hat{\mathcal P}, c\right) \in$ \\$ \mathbb R ^{2 \mathrm{card}\left(\hat{\mathcal{P}}\right)\left(p - \mathrm{card}\left(\hat{\mathcal{P}}\right) \right) + r\left(\hat{\mathcal{P}}\right) 
+ 1}$ such that 

\begin{equation}
    \label{eqn:character}
    {\mathcal{G}}\left(\mathbf x; \hat{\mathcal{P}} \right) = \left\{ \bm{\lambda} = \left(\lambda_1, \lambda_2, \ldots, \lambda_{r\left(\hat{\mathcal{P}}\right)}\right)^\top \in \mathbb R^{r\left(\hat{\mathcal{P}}\right)}  : \mathbf{L} \bm{\lambda} \leqslant \mathbf{g}  \right\}.
\end{equation}
\end{theorem}

\noindent Thus, ${\mathcal{G}}\left(\mathbf x; \hat{\mathcal{P}} \right)$ in \eqref{eqn:conditioning_set} is a convex polytope. The proof of this theorem is provided in  Supplementary Material Section \ref{Appendix:character}.

\section{Computation of $p\left(\mathbf x; \hat{\mathcal{P}}\right)$ in \eqref{eqn:thm_main}}
\label{sec:summary}

We observe that $p\left(\mathbf x; \hat{\mathcal{P}}\right)$ in \eqref{eqn:thm_main} can be written as
\begin{align*}
p\left(\mathbf x; \hat{\mathcal{P}}\right) 
&= \frac{\int_{\left\{\bm{\lambda} \in {\mathcal{G}}\left(\mathbf x; \hat{\mathcal{P}} \right)\right\} \cap \left\{ \prod_{i= 1}^{r\left({\hat{\mathcal P}}\right)}  \left(1 - \lambda_i^2\right) \leqslant  \prod_{i = 1}^{r\left({\hat{\mathcal P}}\right)} \left(1 - \left(\hat\lambda_i^{\hat{\mathcal P}}\left(\mathbf x\right)\right)^2\right)\right\} }f\left(\lambda_1, \lambda_2, \ldots, \lambda_{r\left(\hat{\mathcal{P}}\right)}\right)d\lambda_1d\lambda_2\ldots d\lambda_{r\left(\hat{\mathcal{P}}\right)} }{\int_{\bm{\lambda} \in {\mathcal{G}}\left(\mathbf x; \hat{\mathcal{P}} \right)} f\left(\lambda_1, \lambda_2, \ldots, \lambda_{r\left(\hat{\mathcal{P}}\right)}\right)d\lambda_1d\lambda_2\ldots d\lambda_{r\left(\hat{\mathcal{P}}\right)}}\\
&= \frac{\int_{\left\{\mathbf{L}\bm{\lambda} \leqslant \mathbf g\right\} \cap \left\{ \prod_{i= 1}^{r\left({\hat{\mathcal P}}\right)}  \left(1 - \lambda_i^2\right) \leqslant  \prod_{i = 1}^{r\left({\hat{\mathcal P}}\right)} \left(1 - \left(\hat\lambda_i^{\hat{\mathcal P}}\left(\mathbf x\right)\right)^2\right)\right\} }f\left(\lambda_1, \lambda_2, \ldots, \lambda_{r\left(\hat{\mathcal{P}}\right)}\right)d\lambda_1d\lambda_2\ldots d\lambda_{r\left(\hat{\mathcal{P}}\right)} }{\int_{\mathbf{L}\bm{\lambda} \leqslant \mathbf g} f\left(\lambda_1, \lambda_2, \ldots, \lambda_{r\left(\hat{\mathcal{P}}\right)}\right)d\lambda_1d\lambda_2\ldots d\lambda_{r\left(\hat{\mathcal{P}}\right)}}, \numberthis \label{eqn:probability}
\end{align*}
where $f\left(\lambda_1, \lambda_2, \ldots, \lambda_{r\left(\hat{\mathcal{P}}\right)}\right)$, introduced in \eqref{eqn:density_singular_values}, is the joint density of the canonical correlations if $\hat{\mathcal P}$ is prespecified, and the last equality follows from \eqref{eqn:character}. 

\subsection{Simple computation when $r\left( \hat{\mathcal P}\right) = 1$}\label{sec:mainrp=1}
In Supplementary Material Section \ref{Appendix:singleton_beta}, we show that if $r\left(\hat{\mathcal{P}}\right) = 1$, then $p\left(\mathbf x; \hat{\mathcal{P}}\right) $ in \eqref{eqn:probability} can be written in terms of the cumulative distribution function of the univariate beta distribution:

\begin{proposition}
\label{prop:univariate_beta}
If $r\left(\hat{\mathcal{P}}\right) = 1$, then there exists $g_u = g_u\left(\mathbf S\left(\mathbf x \right),\hat{\mathcal{P}}, c\right) \in [0, 1]$ such that $p\left(\mathbf x; \hat{\mathcal{P}}\right) $ in \eqref{eqn:probability} can be written as function of truncated beta distributions, as follows: 

\begin{equation}
\label{eqn:univariate_beta_proposition}
    p\left(\mathbf x; \hat{\mathcal{P}}\right) = \frac{F(g_u) - F\left(\min\left\{g_u, \left(\hat\lambda_1^{\hat{\mathcal P}}\left(\mathbf x\right)\right)^2\right\}\right)} {F(g_u) - F(0)},
\end{equation}
where $F$ is the cumulative distribution function of a $\mathrm{Beta}\left(\frac{p - 1}{2}, \frac{n - p}{2}\right)$ distribution.
\end{proposition}
In Proposition \ref{prop:univariate_beta}, $c$ is the threshold used in Algorithm \ref{algo:Selection_procedure}. 
\subsection{Numerical integration when $r\left( \hat{\mathcal{P}}\right)$ is small}\label{sec:mainrp_small}
 Evaluating \eqref{eqn:probability} is challenging, because we do not have a closed form expression for the integrals involved.  When $r\left( \hat{\mathcal{P}}\right)$ is a small number that exceeds one, we adopt methods of numerical integration on a convex polytope to evaluate \eqref{eqn:probability}. To evaluate the integral of a function over the convex polytope $\left\{ \bm\lambda : \mathbf L\bm\lambda \leqslant \mathbf g \right\}$, we first partition the polytope into simplices using Delaunay triangulation \citep{lee1980two}, next integrate the corresponding function over those simplices, and finally sum those integrated values. Details of this computation to approximate \eqref{eqn:probability} are provided in Supplementary Material Section \ref{Appendix:geometry_integration}.
\subsection{Monte Carlo approximation when $r\left( \hat{\mathcal P}\right)$ is large}\label{sec:mainrp_MC}
For moderately large values of $r\left(\hat{\mathcal{P}}\right)$, the approach in Section \ref{sec:mainrp_small} is computationally taxing or infeasible. Thus, we resort to Monte Carlo approximation. From \eqref{eqn:probability}, we have that
\begin{equation}
    \label{eqn:MC_base}
    p\left(\mathbf x; \hat{\mathcal{P}}\right) = \frac{\mathbb E_{ H_0^{\hat{\mathcal{P}}}}\left(\mathds{1}\left\{\prod_{i=1}^{r\left({\hat{\mathcal P}}\right)} \left(1 - \lambda_i^{2}\right) \leqslant \prod_{i = 1}^{r\left({\hat{\mathcal P}}\right)} \left(1 - \left(\hat\lambda_i^{\hat{\mathcal P}}\left(\mathbf x\right)\right)^2\right), \:\: \mathbf{L}{\bm\lambda} \leqslant \mathbf{g} \right\}\right)}{\mathbb E_{ H_0^{\hat{\mathcal{P}}}} \left( \mathds{1}\left\{\mathbf{L}{\bm\lambda} \leqslant \mathbf{g}\right\}\right)}.
\end{equation}

\noindent For $b = 1,2,\ldots, B$, we simulate $\bm{\lambda}^{\left(b\right)} = \left(\lambda_1^{\left(b\right)}, \lambda_2^{\left(b\right)}, \ldots, \lambda_{r\left(\hat{\mathcal{P}}\right)}^{\left(b\right)}\right)$ following (ii) in Proposition \ref{prop:density_singular_values}. Using \eqref{eqn:MC_base}, we then approximate $ p\left(\mathbf x; \hat{\mathcal{P}}\right) $  as

\begin{equation}
    \label{eqn:Naive_MC}
\begin{aligned}
p\left(\mathbf x; \hat{\mathcal{P}}\right) &\approx \hat p_{\mathrm{MC}}\left(\mathbf x; \hat{\mathcal{P}}\right) \\ &=\frac{\sum_{b = 1}^{B}\mathds{1}\left\{\prod_{i=1}^{r\left({\hat{\mathcal P}}\right)} \left(1 - \left(\lambda_i^{{\left(b\right)}}\right)^2\right) \leqslant \prod_{i = 1}^{r\left({\hat{\mathcal P}}\right)} \left(1 - \left(\hat\lambda_i^{\hat{\mathcal P}}\left(\mathbf x\right)\right)^2\right), \:\: \mathbf{L}{\bm\lambda^{\left(b\right)}} \leqslant \mathbf{g} \right\}}{    \sum_{b = 1}^{B} \mathds{1}\left\{\mathbf{L}{\bm\lambda^{\left(b\right)}} \leqslant \mathbf{g}\right\}}.
\end{aligned}
\end{equation}

The overall procedure for computing $ p\left(\mathbf x; \hat{\mathcal{P}}\right) $ is summarized in Algorithm \ref{algo:SI} of Supplementary Material Section \ref{Appendix:algorithm}.

\section{Simulation results}
\label{sec:simulation}

\subsection{Type I error under global null}
\label{sec:type_I}
We simulate data with unordered variables from \eqref{eqn:distribution_data} with $\bSigma = \bI$, so that $H_0^{\hat{\mathcal P}}$ holds for all partitions of the variables. We fix $p = 100$ and vary $n \in \left\{1.1p, 1.5p, 2p\right\}$. For each simulated
data set, we compute both the classical p-value $\left( p_{\mathrm{LRT}}\right)$ in \eqref{eqn:base_hyp_simplified} and the selective p-value in \eqref{eqn:selective_pvalue_definition} for the hypothesis $H_0^{\hat{\mathcal{P}}} : \bSigma_{\hat{\mathcal{P}}, \hat{\mathcal{P}}^c} = \mathbf 0$ for a randomly chosen $\hat{\mathcal{P}} \in \mathcal{P}\left(\bS\left(\bx\right)\right)$ with the procedure $\mathcal{P}\left(\cdot\right)$ defined in Section  \ref{subsection:thresholding_procedure} with the threshold $c = 0.2$. We approximate the selective p-value in \eqref{eqn:selective_pvalue_definition} as discussed in Algorithm \ref{algo:SI} of Supplementary Material Section \ref{Appendix:algorithm}. 

We now consider computing the classical p-value in \eqref{eqn:base_hyp_simplified}. Recall from Section \ref{sec:prob_def_hyp} that if $\hat{\mathcal P}$ is not a function of the data, then the test statistic in \eqref{eqn:test_stat_simplified} follows a Wilks' lambda distribution. However, to the best of our knowledge, an exact evaluation of Wilks' lambda distribution is not available in \texttt{R}. Hence we make use of Proposition \ref{prop:density_singular_values} to evaluate the classical p-value $p_{\mathrm{LRT}}$. For $r\left( \hat{\mathcal P}\right) = 1$, (iii) of Proposition \ref{prop:density_singular_values} gives us a closed form evaluation of $ p_{\mathrm{LRT}}$. For $r\left( \hat{\mathcal P}\right) > 1$, following Section \ref{sec:mainrp_MC}, we use a Monte Carlo approach based on (ii) of Proposition \ref{prop:density_singular_values} to approximate $p_{\mathrm{LRT}}$. 

Figure \ref{Fig:global_null_p_total_100} displays the QQ plots of the empirical distribution of the p-values against the Uniform$\left(0, 1\right)$ distribution, over $100,000$ simulated data set. The classical p-value is overly conservative, whereas the proposed selective p-value is well-calibrated.

\begin{figure}[t!]
    \centering
    \includegraphics[height=4.2in]{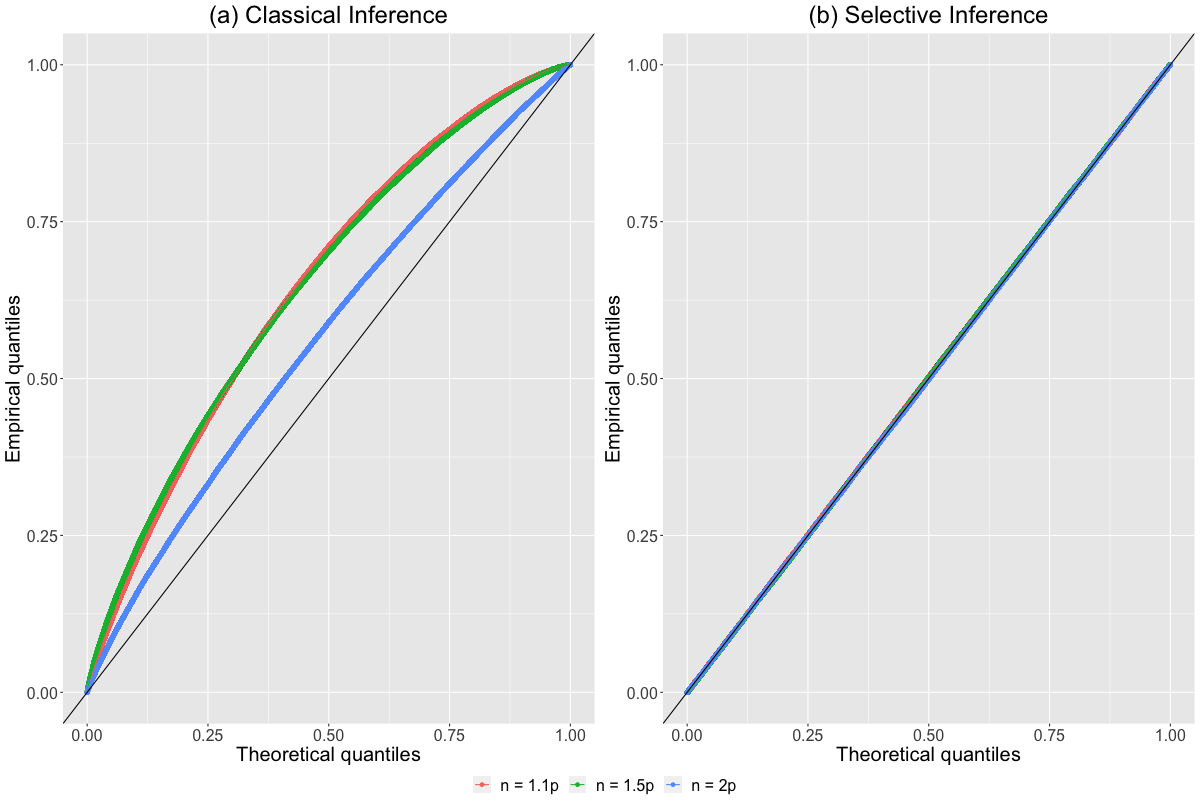}
    \caption{For 100,000 draws from \eqref{eqn:distribution_data}, with $\bm\Sigma = \mathbf I_{100}, n \in \left\{1.1p, 1.5p, 2p\right\}$, QQ-plots are displayed for: (a) the ``classical'' p-value in  \eqref{eqn:base_hyp_simplified}; (b) the ``selective'' p-value in \eqref{eqn:selective_pvalue_definition}, approximated as described in Section \ref{sec:summary}. }\label{Fig:global_null_p_total_100}
\end{figure}

\subsection{Power at $\alpha = 0.05$}
\label{sec:power}
We simulate data with unordered variables from \eqref{eqn:distribution_data}, with $p = 100$, $n \in \left\{1.1p, 1.5p, 2p\right\}$. For each simulated data set, we generate a random $p \times p$ positive definite matrix $\bSigma$ and a corresponding threshold $c = c\left(\bSigma\right)$ as described in Supplementary Material Section \ref{Appendix:Simulation}, and we randomly select $\hat{\mathcal{P}} \in \mathcal{P}\left(\bS\left(\bx\right)\right)$ with $\mathcal{P}\left(\cdot\right)$ defined in Section \ref{subsection:thresholding_procedure}. We test $H_0^{\hat{\mathcal{P}}} : \bSigma_{\hat{\mathcal{P}}, \hat{\mathcal{P}}^c} = \mathbf 0$ at significance
level $\alpha = 0.05$. Motivated by the test statistic in \eqref{eqn:lrt}, we define the effect size based on its population counterpart as $\Delta := 1 - \frac{\left|\bSigma\right|}{\left|\bSigma_{{\hat{\mathcal{P}}}, {\hat{\mathcal{P}}}}\right|\left|\bSigma_{{\hat{\mathcal{P}}^c}, {\hat{\mathcal{P}}^c}}\right|}$, and consider the
probability of rejecting $H_0^{\hat{\mathcal{P}}} : \bSigma_{\hat{\mathcal{P}}, \hat{\mathcal{P}}^c} = \mathbf 0$ as a function of $\Delta$ over $1,000,000$ simulated data sets. Figure \ref{Fig:power_total_100} shows that the proposed selective inference approach has higher power than classical inference for all values of $\Delta$ and $n$. 

\begin{figure}[t!]
    \centering
    \includegraphics[height=4.2in]{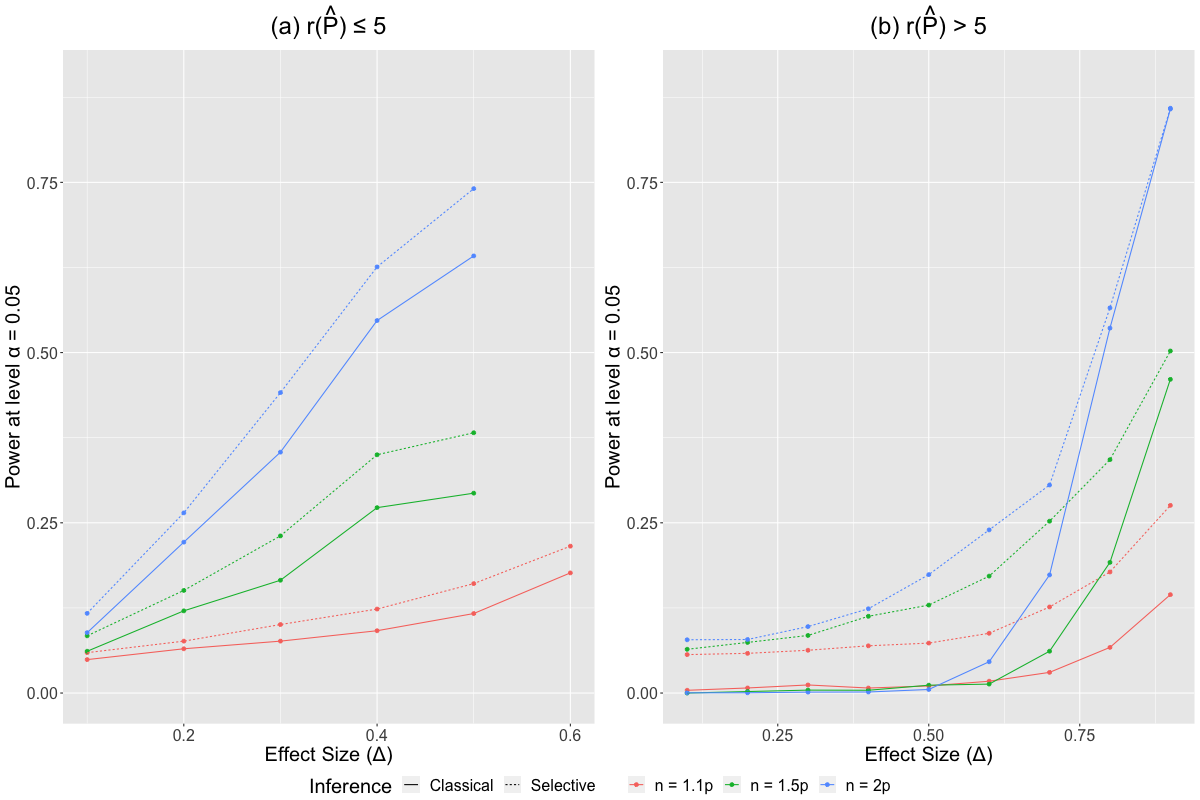}
    \caption{Power estimates as a function of the effect-size $\Delta$ for 1,000,000 draws from the simulation setup described in Section \ref{sec:power}, when (a) $ r\left(\hat{\mathcal{P}}\right) \leqslant 5$, (b) $r\left(\hat{\mathcal{P}}\right) > 5$, where $r\left(\hat{\mathcal{P}}\right) = \min\left\{\mathrm{card}\left({\hat{\mathcal{P}}}\right), \mathrm{card}\left(\hat{\mathcal{P}}^c\right) \right\}$. The x-axis displays bins of $\Delta$ of size $0.1$; only bins with more than $100$ simulated datasets are displayed. }\label{Fig:power_total_100}
\end{figure}

\section{Application to gene co-expression networks}
\label{sec:gene_coexpression}

A gene co-expression network is an undirected graph in which the nodes represent genes and the edges represent pairs of genes that are co-expressed, in the sense of e.g. Pearson correlation \citep{guttman2011lincrnas}. Although gene co-expression networks do not confer information about causality, they are frequently used as a starting point for investigating gene regulation, functional enrichment, and hub genes \citep{van2018gene}. Additionally, they have proven  useful in identifying disease genes through the  \textit{guilt-by-association} principle, enabling a better understanding of disease origin \citep{alsina2014narrow} and progression \citep{chaussabel2008modular}. For a detailed review of gene co-expression networks, we refer readers to \citet{saelens2018comprehensive} and references therein.

In this section, we focus on the gene expression data corresponding to \textit{E. coli} and  \textit{S. cerevisiae} from the DREAM5 network inference challenge \citep{marbach2012wisdom}. We will test for correlation between sets of genes that are known to be correlated based on outside biological knowledge. Our results indicate that our proposed selective inference approach has higher power than the classical inference approach. 

\subsection{\textit{Escherichia coli} and \textit{Saccharomyces cerevisiae} data sets}
The DREAM5 \textit{E. coli} data set consists of pre-processed gene expression measurements for $p = $ 4,297 genes on $n = $ 805 microarray chips \citep{marbach2012wisdom}. Furthermore, the \textit{E. coli} gene co-expression network is very well understood, and known edges can be found in publicly available databases, e.g. RegulonDB \citep{gama2010regulondb} and EcoCyc \citep{keseler2010ecocyc}. As part of the DREAM5 challenge, a highly curated co-expression network, based on this outside biological information, was made available. We treat this as a ``ground truth'' co-expression network in what follows. 

The DREAM5 \textit{S. cerevisiae} data set consists of $p = $ 5,667 pre-processed gene expression measurements on $n = $ 536 microarrays. While the \textit{S. cerevisiae} network is not as well understood as the \textit{E. coli} network, the DREAM5 organizers combined data from 16 sources to arrive at a single co-expression network \citep{macisaac2006improved}. We treat this network as the ``ground truth'' in what follows.

\subsection{Power to detect correlation between sets of genes}
For each data set, we identify the gene with the most edges in the ground truth co-expression network, and restrict our analysis to that hub gene and $99$ genes randomly sampled from the genes to which it is connected. This yields a set of $p = 100$ genes that form a single connected component in the ground truth co-expression network. By construction, for any $\hat{\mathcal P} \in \left\{1,2,\ldots, p \right\}$, the null hypothesis $H_0^{\hat{\mathcal{P}}} : \bSigma_{\hat{\mathcal{P}}, \hat{\mathcal{P}}^c} = \mathbf 0$ will not hold.

Next, we sample $n \in \left\{1.2p, 1.4p, 1.5p, 2p\right\}$ microarrays, and let $\bx$ denote the resulting $n \times p$ data set. We then compute $\mathcal{P}\left(\bS\left(\bx\right)\right)$ using the procedure in Section \ref{subsection:thresholding_procedure}.  To choose the threshold $c$, we apply the procedure described in Supplementary Material Section \ref{Appendix:Simulation} to the sample covariance matrix computed on the held-out microarrays $\left(\text{which we write as }\hat{\bS}_{holdout}\right)$, i.e. those not included in $\bx$. Since the microarrays are independent, this threshold $c$ is independent of $\bx$. For a randomly selected $\hat{\mathcal{P}} \in \mathcal{P}\left(\bS\left(\bx\right)\right)$, we test $H_0^{\hat{\mathcal{P}}} : \bSigma_{\hat{\mathcal{P}}, \hat{\mathcal{P}}^c} = \mathbf 0$ at significance level $\alpha = 0.05$, where $\bSigma$ denotes the population covariance matrix of the $p$ genes. We repeat this entire procedure 10,000 times for different random sets of $99$ neighbors of the hub gene.

Figures \ref{Fig:power_ecoli}(a) and \ref{Fig:power_scerv}(a) display the probability of rejecting $H_0^{\hat{\mathcal{P}}} : \bSigma_{\hat{\mathcal{P}}, \hat{\mathcal{P}}^c} = \mathbf 0$ as a function of $\hat\Delta := 1 - \frac{\left|\hat{\bS}_{holdout}\right|}{\left|\hat{\bS}_{{holdout}_{\hat{\mathcal{P}}, \hat{\mathcal{P}}}}\right|\left|\hat{\bS}_{{holdout}_{\hat{\mathcal{P}}^c, \hat{\mathcal{P}}^c}}\right|}$. In this setting, by construction, the $p = 100$ genes have extremely high correlation. For any $\hat{\mathcal P} \in \mathcal P \left(\bS \left( \bx\right) \right)$, the conditioning set in \eqref{eqn:thm_main} deals with a modified version $\bS' \left( \bx, \hat{\mathcal P}, \cdot\right)$ of the sample covariance matrix $\bS \left( \bx\right)$ with a perturbed off-diagonal block (recall Theorem \ref{thm:main}). Because the genes are highly correlated, for perturbations under $H_0^{\hat{\mathcal P}}$, with very high probability, we will have $\mathcal P \left(\bS' \left( \bx, \hat{\mathcal P}, \cdot\right)\right) = \mathcal P \left( \bS \left( \bx \right)\right)$. Hence the classical and selective approaches yield almost identical results.

To investigate the performance of the classical and selective approaches when the correlation between genes is lower, we add Gaussian white noise to each gene's expression, as follows: 

\begin{equation}
    \label{eqn:noise_realdata}
     \tilde{\mathbf x}_j = \mathbf x_j + \eta \bm\varepsilon_j, \:\:\: \bm\varepsilon_j \sim N\left(\mathbf 0, \widehat{\mathrm{var}}\left(\mathbf x_j\right)\mathbf I\right), \:\: j = 1, 2,\ldots, p.
\end{equation}
Here $\mathbf x_j \in \mathbb R^n$ is the expression vector for the $j^{th}$ gene, i.e. the $j^{th}$ column of $\bx$. We vary the noise level $\eta \in \left\{0.5, 3\right\}$. In Figures \ref{Fig:power_ecoli} and \ref{Fig:power_scerv}, we see that the power of both the classical and selective approaches decreases as $\eta$ increases. Moreover, as $\eta$ increases, the correlation between the genes decreases. Hence for perturbations under $H_0^{\hat{\mathcal P}}$, the probability of the event $\mathcal P \left(\bS' \left( \bx, \hat{\mathcal P}, \cdot\right)\right) = \mathcal P \left( \bS \left( \bx \right)\right)$ decreases. As a result, the difference between the classical and selective approaches increases, and the selective approach becomes more powerful than the classical approach. 

\begin{figure}[t!]
    \centering
    \includegraphics[height=2.2in]{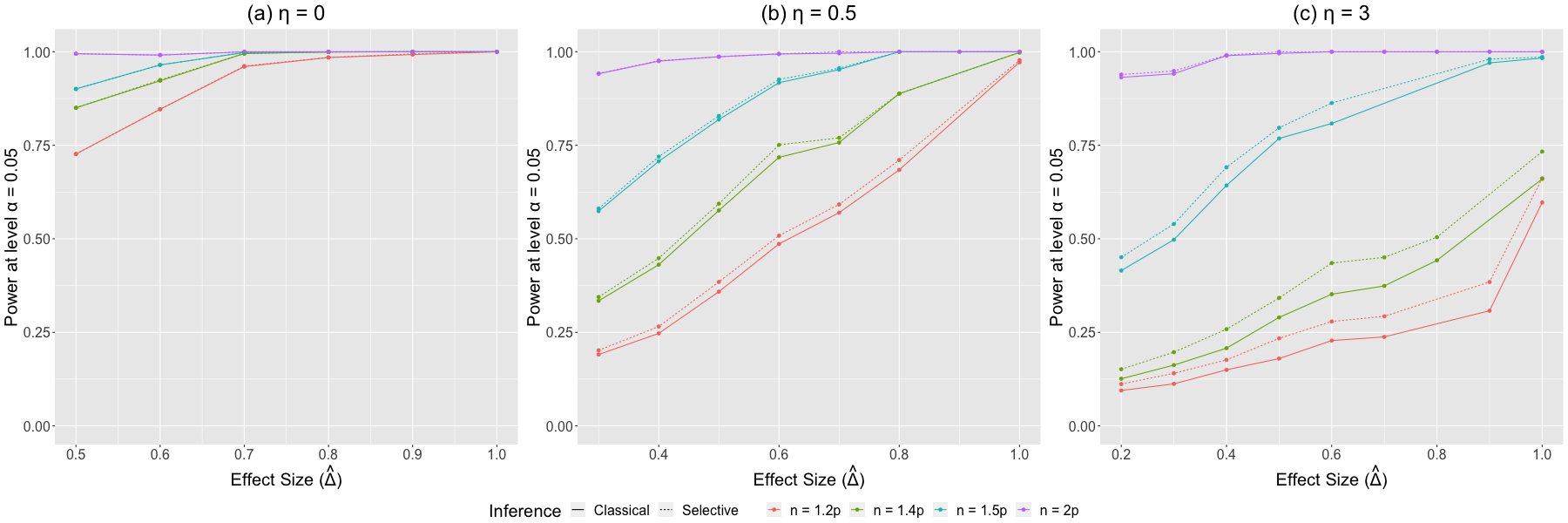}
    \caption{Power estimates as a function of $\hat\Delta$ for \textit{E. coli} for three values of $\eta$ in \eqref{eqn:noise_realdata}. (a) $ \eta = 0$, (b) $ \eta = 0.5$, (c) $ \eta = 3$  in \eqref{eqn:noise_realdata}. We only display bins of $\hat\Delta$ with more than 100 observations.}\label{Fig:power_ecoli}
\end{figure}

\begin{figure}[t!]
    \centering
    \includegraphics[height=2.2in]{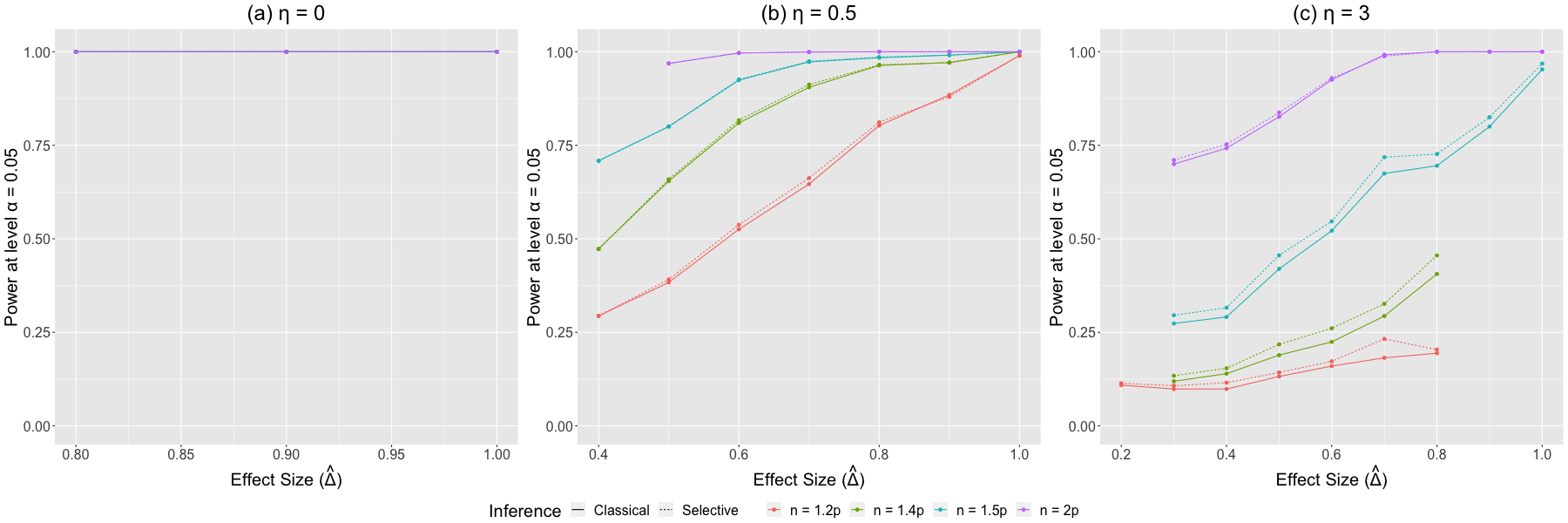}
    \caption{Power estimates as a function of $\hat\Delta$ for \textit{S. cerevisiae} for three values of $\eta$ in \eqref{eqn:noise_realdata}. (a) $ \eta = 0$, (b) $ \eta = 0.5$, (c) $ \eta = 3$  in \eqref{eqn:noise_realdata}. We only display bins of $\hat\Delta$ with more than 100 observations.}\label{Fig:power_scerv}
\end{figure}

\section{Discussion}
\label{sec:discussion}
Our proposed approach for testing dependence between data-driven groups of Gaussian variables requires the number of features to be less than the number of observations, i.e. $p < n$, in order for the test statistic to be defined. Future work could combine tools from high-dimensional inference with selective inference to account for selection in covariance matrices with $p \gg n$.

The normality assumption on the variables is needed to obtain Proposition \ref{prop:density_singular_values}, but is restrictive in practice. For example, RNA sequencing data involve nonnegative counts, which are better modeled with a Poisson \citep{witten2011classification} or negative binomial \citep{risso2018general} distribution. Extending the proposed selective inference approach to accommodate a more general class of distributions requires further investigation.  

An \texttt{R} package to implement the proposed approach called \texttt{independencepvalue} is available at \href{https://arkajyotisaha.github.io/independencepvalue-project}{https://arkajyotisaha.github.io/independencepvalue-project}.

\bibliographystyle{apalike}

\bibliography{arXiv_AS}

\pagebreak

\section*{Supplementary Materials for ``Inferring independent sets of Gaussian variables after thresholding correlations"}
\renewcommand\theequation{S\arabic{equation}}
\renewcommand\thesection{S\arabic{section}}
\renewcommand\thealgorithm{S\arabic{algorithm}}
\setcounter{equation}{0}
\setcounter{section}{0}
\setcounter{algorithm}{0}

\section{Brief review of canonical correlation analysis}
\label{Appendix:CCA}
Here we briefly review canonical correlation analysis, as it paves the way for the selective inference method proposed in this paper. We largely follow the notation used in \citet{anderson1962introduction}. 

We partition the random vector $X \sim N\left(0, \bm\Sigma\right)$ of length $p$ into two vectors $X_1$ and  $X_2$ of lengths $p_1$ and $p_2 = p - p_1$ respectively, so that $X^\top = \left(X_1^\top\:\:\: X_2^\top\right)$. The corresponding covariance matrix can be partitioned as
$$
\bSigma = \begin{pmatrix}
\bSigma_{11} & \bSigma_{12}\\
\bSigma_{21} & \bSigma_{22}
\end{pmatrix},
$$
where $\bm\Sigma$ is of dimension $p \times p$, $\bm\Sigma_{12}$ is of dimension $p_1 \times p_2$, etc. Without loss of generality, we assume $p_1 \leqslant p_2$. We define the population canonical correlations $\lambda_1, \lambda_2, \ldots, \lambda_{p_1}$ between $X_1$ and $X_2$, and their corresponding population canonical vectors $\bA^* = \left(\alpha_1^*, \alpha_2^*, \ldots, \alpha_{p_1}^* \right) \in \mathbb R^{p_1 \times p_1}$ and $\bm\Gamma^*  = \left(\gamma_1^*, \gamma_2^*, \ldots, \gamma_{p_1}^* \right) \in \mathbb R ^{p_2 \times p_1}$, sequentially for $r = 1, 2, \ldots, p_1$:
\begin{subequations}
\begin{equation}
\label{eqn:theoretical_CCA}
    \begin{aligned}
    \left(\bm\alpha_r^*, \bm\gamma_r^*\right) :=& \argmaxA_{\substack{\bm\alpha_r^\top \bm\Sigma_{11} \bm\alpha_r =1\\
    \bm\gamma_r^\top \bm\Sigma_{22} \bm\gamma_r =1\\
    \bm\alpha_r^\top \bm\Sigma_{11} \bm\alpha_j^* = 0; \forall j < r\\
    \bm\gamma_r^\top \bm\Sigma_{22} \bm\gamma_j^* = 0; \forall j < r}} \bm\alpha_r^\top \bm\Sigma_{12} \bm\gamma_r;\:\:\:\:\:\:\:\:\:\:\:\:\: \lambda_r := \left(\bm\alpha_r^*\right)^\top \bm\Sigma_{12} \bm\gamma_r^*.
\end{aligned}
\end{equation}
The above can be solved through an SVD on $\bSigma_{11}^{-\frac{1}{2}}\bSigma_{12}\bSigma_{22}^{-\frac{1}{2}} = \mathbf A \bm\Lambda \bm\Gamma^\top$, by taking $\mathbf A^* = \bSigma_{11}^{-\frac{1}{2}}\mathbf A, \bm\Gamma^* = \bSigma_{22}^{-\frac{1}{2}}\bm\Gamma$, and $\bm\lambda_r = \mathrm{diag}\left( \bm\Lambda\right)$ \citep{jendoubi2019whitening}. 
Next, we define the sample canonical correlations and sample canonical vectors for a realization $\bx \in \mathbb R^{n \times p}$ of $\bX$ in \eqref{eqn:distribution_data}. We assume $\mathrm{rank}\left(\bx \right)= p$, which holds with probability $1$, since we assume $\bm\Sigma$ is positive definite. We let $\bS = \begin{pmatrix}
\bS_{11} & \bS_{12}\\
\bS_{12}^\top & \bS_{22}
\end{pmatrix}$ denote the sample covariance of $\bx = \left( \bx_1^\top\:\:\: \bx_2^\top\right)^\top$, suitably partitioned. We now define the ``whitened'' versions of $\bx_1$ and $\bx_2$, i.e.
$
{\bx}^W_1 = \bS_{11}^{-\frac{1}{2}}\bx_1$ and ${\bx}^W_2 = \bS_{22}^{-\frac{1}{2}}\bx_2.
$
Let $\whitenS= \begin{pmatrix}
\whitenS_{11} & \whitenS_{12}\\
\left(\whitenS_{12}\right)^\top & \whitenS_{22}
\end{pmatrix}$ denote the sample covariance matrix of ${\bx}^W = \left(\left({\bx}^W_1\right)^\top\:\:\: \left({\bx}^W_2\right)^\top\right)^\top$, suitably partitioned. It follows that
$
\whitenS_{11} = \bI_{p_1}; \: \whitenS_{22} = \bI_{p_2}, \text{ and }\whitenS_{12} = \bS_{11}^{-\frac{1}{2}} \bS_{12} \bS_{22}^{-\frac{1}{2}}$. The compact SVD of $\whitenS_{12}$ is
\begin{equation}
\label{eqn:SVD}
    \begin{aligned}
        \whitenS_{12} \overset{SVD}{=} \hat{\bA}\:\hat{\bm{\Lambda}}\:\hat{\bm{\Gamma}}^\top,
    \end{aligned}
\end{equation}
\end{subequations}
where $\hat{\bm{\Lambda}}$ is a square diagonal matrix of $p_1$ positive singular values, arranged in non-increasing order, and $\hat{\bA}$ and $\hat{\bm{\Gamma}}$ are the $p_1 \times p_1$ and $p_2 \times p_1$ matrices of  left and right singular vectors, respectively. We now derive \eqref{eqn:test_stat_simplified}. We can rewrite the classical LRT statistic in \eqref{eqn:lrt} by noting that
$$
\begin{aligned}
    \left| \bS\right| &= \left| \bS_{22}\right|\left| \bS_{11} - \bS_{12}\bS_{22}^{-1}\bS_{21}\right|\\
    &= \left| \bS_{22}\right|\left| \bS_{11}\right|\left| \bI - \bS_{11}^{-\frac{1}{2}}\bS_{12}\bS_{22}^{-\frac{1}{2}}\bS_{22}^{-\frac{1}{2}}\bS_{21}\bS_{11}^{-\frac{1}{2}}\right|\\
    &= \left| \bS_{11}\right|\left| \bS_{22}\right|\left| \bI - \whitenS_{12}\left( \whitenS_{12}\right)^\top\right|\\
    &= \left| \bS_{11}\right|\left| \bS_{22}\right|\left| \bI - \hat{\bm\Lambda}^2\right|,
\end{aligned}
$$
where in the last equality, we have used that $\hat{\mathbf A}$ is a square orthogonal matrix.

\section{Proof of Proposition \ref{prop:density_singular_values}}
\label{Appendix:density}
For any subset $\mathcal P$ of $\{1,2,\ldots, p \}$, the sample canonical correlations are the diagonal elements of $\hat{\bm\Lambda}^{\mathcal P}$, denoted by $\left\{\hat{\lambda}_i^{\mathcal P} \right\}_{i = 1}^{r\left({\mathcal P}\right)}$. In this section, we will omit the ${\mathcal P}$ superscript and write $\left\{\hat{\lambda}_i\right\}_{i = 1}^{r\left({\mathcal P}\right)}$.  
\subsection{Proof of Proposition \ref{prop:density_singular_values}(i)}
\begin{proof}
The joint density of $\left(\hat\lambda_1^2, \hat\lambda_2^2, \ldots, \hat\lambda_{r\left(\mathcal P\right)}^2\right)$ under $H_0^\mathcal{P}$ takes the form 
\small
\begin{equation}
\label{eqn:density_singular_values_squared}
   \begin{aligned}
    f\left(\hat\lambda_1^2 = \gamma_1, \hat\lambda_2^2 = \gamma_2, \ldots, \hat\lambda_{r\left(\mathcal P\right)}^2 = \gamma_{r\left(\mathcal P\right)}\right) =& \pi^{\frac{1}{2}\left(r\left(\mathcal P\right)\right)^2}\frac{\Gamma_{r\left(\mathcal P\right)}\left(\frac{n}{2}\right)}{\Gamma_{r\left(\mathcal P\right)}\left(\frac{n-p + r\left(\mathcal P\right)}{2}\right)\Gamma_{p}\left(\frac{1}{2}r\left(\mathcal P\right)\right)\Gamma _{p}\left(\frac{1}{2}\left(p - r\left(\mathcal P\right)\right)\right)} \times\\ 
    &\prod_{i = 1}^{r\left(\mathcal P\right)}\left\{\gamma_i^{\frac{1}{2}\left(\left|p - 2\mathrm{card}\left({\mathcal{P}}\right)\right| - 1\right)}\left(1 - \gamma_i\right)^{\frac{1}{2}\left(n - p - 2\right)} \right\}\prod_{i < j}\left(\gamma_i - \gamma_j\right)
   \end{aligned}
\end{equation}
\normalsize
\noindent by (13) in Section 13.4 in \citet{anderson1962introduction}. To obtain $f\left(\hat\lambda_1, \hat\lambda_2, \ldots, \hat\lambda_{r\left(\mathcal P\right)}\right)$ in \eqref{eqn:density_singular_values}, we perform a change of variables. We have 
\begin{equation}
    \label{eqn:jacobin_1}
    f\left(\hat\lambda_1, \hat\lambda_2, \ldots, \hat\lambda_{r\left(\mathcal P\right)}\right) = f\left(\hat\lambda_1^2, \hat\lambda_2^2, \ldots, \hat\lambda_{r\left(\mathcal P\right)}^2\right) \left|\mathbf J\right|,
\end{equation}
where $\mathbf J$ is the $r\left(\mathcal P\right) \times r\left(\mathcal P\right)$ Jacobian matrix of the transformation, with $\mathbf J_{i,j} = \frac{\partial \hat\lambda_i^2}{\partial\hat\lambda_j}$ for $ i,j \in \left\{1,2,\ldots, r\left(\mathcal P\right) \right\}$. Simplifying, we have $\left|\mathbf J\right| = 2^{r\left(\mathcal P\right)} \prod_{i = 1}^{r\left(\mathcal P\right)} \hat\lambda_i$. Substituting this into \eqref{eqn:jacobin_1} yields
\begin{equation}
    \begin{aligned}
     f\left(\hat\lambda_1, \hat\lambda_2, \ldots, \hat\lambda_{r\left(\mathcal P\right)}\right) = & \pi^{\frac{1}{2}\left(r\left(\mathcal P\right)\right)^2}\frac{\Gamma_{r\left(\mathcal P\right)}\left(\frac{n}{2}\right)}{\Gamma_{r\left(\mathcal P\right)}\left(\frac{n-p + r\left(\mathcal P\right)}{2}\right)\Gamma_{p}\left(\frac{1}{2}r\left(\mathcal P\right)\right)\Gamma _{p}\left(\frac{1}{2}\left(p - r\left(\mathcal P\right)\right)\right)} \times\\ 
    &\prod_{i = 1}^{r\left(\mathcal P\right)}\left\{\hat\lambda_i^{\left(\left| p - 2\mathrm{card}\left({\mathcal{P}}\right)\right| - 1\right)}\left(1 - \hat\lambda_i^2\right)^{\frac{1}{2}\left(n - p - 2\right)} \right\}\prod_{i < j}\left(\hat\lambda_i^2 - \hat\lambda_j^2\right) \times 2^{r\left(\mathcal P\right)} \prod_{i = 1}^{r\left(\mathcal P\right)} \hat\lambda_i,
    \end{aligned}
\end{equation}
which completes the proof of (i) of Proposition \ref{prop:density_singular_values}. 
\end{proof}

\subsection{Proof of Proposition \ref{prop:density_singular_values}(ii)}
\begin{proof}
The result follows by applying various arguments from \citet{anderson1962introduction}. For clarity, we extract the relevant parts here. For notational ease, we will write $\mathbf S\left( \bx\right)$ as $\mathbf S$. 

\begin{enumerate}[1)]
    \item Applying the argument in (48)-(52) of Section 12.2 of \citet{anderson1962introduction} to the sample covariance matrix, we have that the squared canonical correlations 
    $\left\{\hat{\lambda}_i^2 \right\}_{i = 1}^{r\left({\mathcal P}\right)}$ are the roots, with respect to $\gamma$, of
\begin{equation}
    \label{eqn:MC_simulation_derivation_short_1}
    \begin{aligned}
     f\left(\gamma\right)=\left| \bS_{{\mathcal{P}}, {\mathcal{P}}^c} \left( \bS_{{\mathcal{P}}^c, {\mathcal{P}}^c}\right)^{-1} \bS_{{\mathcal{P}}^c, {\mathcal{P}}}  - \gamma \bS_{{\mathcal{P}}, {\mathcal{P}}} \right|.
    \end{aligned}
\end{equation}

\item By Corollary 4.3.2 of \citet{anderson1962introduction}, if $\bSigma_{{\mathcal{P}}, {\mathcal{P}}^c} = 0$, then $$
\begin{aligned}
          \bW &:= n\left( \bSigma_{{\mathcal{P}}, {\mathcal{P}}}\right)^{-\frac{1}{2}}\bS_{{\mathcal{P}}, {\mathcal{P}}^c} \left( \bS_{{\mathcal{P}}^c, {\mathcal{P}}^c}\right)^{-1} \bS_{{\mathcal{P}}^c, {\mathcal{P}}}\left( \bSigma_{{\mathcal{P}}, {\mathcal{P}}}\right)^{-\frac{1}{2}} \sim \mathrm{Wishart}_{r\left({\mathcal{P}}\right)}\left(\bI, p - r\left({\mathcal{P}}\right)\right),\\
\bT &:= n\left( \bSigma_{{\mathcal{P}}, {\mathcal{P}}}\right)^{-\frac{1}{2}}\left(\bS_{{\mathcal{P}}, {\mathcal{P}}} -  \bS_{{\mathcal{P}}, {\mathcal{P}}^c} \left( \bS_{{\mathcal{P}}^c, {\mathcal{P}}^c}\right)^{-1} \bS_{{\mathcal{P}}^c, {\mathcal{P}}}\right)\left( \bSigma_{{\mathcal{P}}, {\mathcal{P}}}\right)^{-\frac{1}{2}} \sim \mathrm{Wishart}_{r\left({\mathcal{P}}\right)}\left(\bI, n - 1 - p + r\left({\mathcal{P}}\right)\right),
\end{aligned}
$$ and $\bW$ and $\bT$  are independent.

\item Equations 7-10 of Section 13.2 of \citet{anderson1962introduction} establish that the values of $\gamma$ that satisfy $\left| \bW - \gamma \left(\bT + \bW\right)\right| = 0$ are given by $\frac{\Psi_i \left( \bW\bT^{-1}\right)}{1 + \Psi_i \left( \bW\bT^{-1}\right)}$, where $ \Psi_i \left( \bW\bT^{-1}\right)$ are the eigenvalues of $\bW\bT^{-1}$. 

\end{enumerate}

By 1), $\hat{\lambda}_1^2, \hat{\lambda}_2^2, \ldots, \hat{\lambda}_{r\left( \mathcal P\right)}^2$ are given by roots of \eqref{eqn:MC_simulation_derivation_short_1}. Pre- and post-multiplying \\ $ \bS_{{\mathcal{P}}, {\mathcal{P}}^c} \left( \bS_{{\mathcal{P}}^c, {\mathcal{P}}^c}\right)^{-1} \bS_{{\mathcal{P}}^c, {\mathcal{P}}}  - \gamma \bS_{{\mathcal{P}}, {\mathcal{P}}}$ in \eqref{eqn:MC_simulation_derivation_short_1} by $\sqrt{n}\bSigma_{{\mathcal{P}}, {\mathcal{P}}}^{-\frac{1}{2}}$ gives that $\hat{\lambda}_1^2, \hat{\lambda}_2^2, \ldots, \hat{\lambda}_{r\left( \mathcal P\right)}^2$ are the values of $\gamma$ that solve
$$
\left| \bW - \gamma \left(\bT + \bW\right)\right| = 0,
$$
where $\bW$ and $\bT$ are defined in 2).  Applying 3) completes the proof of (ii) of Proposition \ref{prop:density_singular_values}.
\end{proof}

\subsection{Proof of Proposition \ref{prop:density_singular_values}(iii)}
\begin{proof}
To prove Proposition \ref{prop:density_singular_values}(iii), we start with a result whose proof follows from (5) in Section 4.4.1 of \citet{anderson1962introduction}. 
\begin{proposition}
\label{prop:multiple_cor_coeff}
The squared sample multiple correlation coefficient for regression for a response $\mathbf h$ on a set of variables $\mathcal Q$ is given by $
\frac{\mathbf s_{\mathcal Q, \mathbf h}^\top \left(\mathbf S_{\mathcal Q, \mathcal Q}\right)^{-1} \mathbf s_{\mathcal Q, \mathbf h}}{\mathrm{var}\left(\mathbf h \right)}
$, where $s_{\mathcal Q, \mathbf h}$ denotes the vector of covariances between $\mathbf h$ and the variables in $\mathcal Q$, and $\mathbf S_{\mathcal Q, \mathcal Q}$ denotes the covariance matrix of the variables in $\mathcal Q$.
\end{proposition}

Without loss of generality, we assume $\mathrm{card}\left({\mathcal{P}}\right) = p - 1$. 
First we establish that the canonical correlation is the multiple correlation coefficient between the variable in ${\mathcal{P}}^c$ and the remaining $p - 1$ variables in ${\mathcal{P}}$. We write $\mathbf S\left( \bx\right)$ as $\mathbf S$ and denote the vectors $\whitenS_{{\mathcal{P}}, {\mathcal{P}}^c} $ and $\bS_{{\mathcal{P}}, {\mathcal{P}}^c}$ as $\whitens_{{\mathcal{P}}, {\mathcal{P}}^c} $ and $ \bs_{{\mathcal{P}}, {\mathcal{P}}^c}$, respectively, and the scalar $\bS_{{\mathcal{P}}^c, {\mathcal{P}}^c}$ as $s_{{\mathcal{P}}^c, {\mathcal{P}}^c}$. Substituting these in \eqref{eqn:CCA} we have

$$   \whitens_{{\mathcal{P}}, {\mathcal{P}}^c}  =  \frac{ \left(\bS_{{{\mathcal{P}}}, {{\mathcal{P}}}}\right)^{-\frac{1}{2}}\bs_{{\mathcal{P}}, {\mathcal{P}}^c}}{\sqrt{s_{{\mathcal{P}}^c, {\mathcal{P}}^c}}}.
$$
Moreover, $\hat{\bm\Lambda}^{{\mathcal{P}}}$ is a scalar, which we denote as $\hat\lambda_1$. From \eqref{eqn:CCA}, $\hat\lambda_1$ is precisely the singular value of the vector $ \whitens_{{\mathcal{P}}, {\mathcal{P}}^c} $, which is $\left\|\whitens_{{\mathcal{P}}, {\mathcal{P}}^c}\right\|_2$. Hence we have
$$
 \left(\hat\lambda_1\right)^2 = \|  \whitens_{{\mathcal{P}}, {\mathcal{P}}^c}\|_2^2 =  \left( \whitens_{{\mathcal{P}}, {\mathcal{P}}^c}\right)^\top \whitens_{{\mathcal{P}}, {\mathcal{P}}^c} = \frac{\bs_{{\mathcal{P}}, {\mathcal{P}}^c}^\top\left(\bS_{{{\mathcal{P}}}, {{\mathcal{P}}}}\right)^{-1} \bs_{{\mathcal{P}}, {\mathcal{P}}^c}}{s_{{\mathcal{P}}^c, {\mathcal{P}}^c}}.
$$
Using Proposition \ref{prop:multiple_cor_coeff}, $\frac{\bs_{{\mathcal{P}}, {\mathcal{P}}^c}^\top\left(\bS_{{{\mathcal{P}}}, {{\mathcal{P}}}}\right)^{-1} \bs_{{\mathcal{P}}, {\mathcal{P}}^c}}{s_{{\mathcal{P}}^c, {\mathcal{P}}^c}}$ is the squared multiple correlation coefficient for the regression of the single variable in ${\mathcal{P}}^c$ onto the remaining $p - 1$ variables. Under the null, the squared sample multiple correlation coefficient follows a beta distribution with parameters $\frac{p - 1}{2}$ and $\frac{n - p}{2}$ \citep{ali2002null}. This completes the proof of (iii) of Proposition \ref{prop:density_singular_values}.
\end{proof}

\section{Proof of Theorem \ref{thm:main}}
\label{Appendix:Main_theorem}
In this section, we abbreviate  $H_0^{{\mathcal{R}}}, \hat{\mathbf{A}}^{{\mathcal{R}}}\left(\cdot\right),  \hat{\bm\Lambda}^{{\mathcal{R}}}\left(\cdot\right)$, and  $\hat{\mathbf{\Gamma}}^{{\mathcal{R}}}\left(\cdot\right)$ (defined in Section \ref{sec:prob_def})  by omitting the superscript. Without loss of generality, we assume that $\mathcal{R} = \left\{1,2,\ldots, p_1 \right\}$, $\mathcal{R}^c = \left\{p_1+1 , \ldots, p \right\}$, and $p_1 \leqslant p/2$. Define $p_2 =  p - p_1$. Let $\hat{\lambda}_1\left(\bX\right), \hat{\lambda}_2\left(\bX\right), \ldots, \hat{\lambda}_{p_1}\left(\bX\right)$ be the diagonal elements of $\hat{\bm\Lambda}\left(\bX\right)$. 

We let $\bS\left(\bX\right) = \begin{pmatrix}
\bS_{11}\left(\bX\right) & \bS_{12}\left(\bX\right)\\
\bS_{12}^\top\left(\bX\right) & \bS_{22}\left(\bX\right)
\end{pmatrix}$ denote the covariance matrix of $\bX = \left(\bX_1\:\:\: \bX_2 \right)$, suitably partitioned, where $\bX_1$ and $\bX_2$ are defined as the submatrices of $\bX$ with columns in $\mathcal R$ and ${\mathcal{R}}^c$, respectively. Equipped with this notation, we first state a technical lemma whose proof we defer to Section \ref{Appendix:independence_under_null_lemma}.

\begin{lemma}
\label{lemma:independence_under_null}
Under $H_0: \bm{\Lambda} = \mathbf 0$, $\hat{\bm{\Lambda}}\left(\mathbf X\right) $ is independent of $\left(\mathbf{S}_{11}\left(\mathbf{X}\right), \mathbf{S}_{22}\left(\mathbf{X}\right), \hat{\mathbf{A}}\left(\mathbf{X}\right), \hat{\bm{\Gamma}}\left(\mathbf{X}\right)\right)$.
\end{lemma}

By \eqref{eqn:CCA}, $\hat{\mathbf{A}}\left(\mathbf X\right) \hat{\bm{\Lambda}}\left(\bX\right) \left(\hat{\mathbf{\Gamma}}\left(\mathbf X\right)\right)^\top$ is the SVD of $\left(\mathbf S_{11}\left(\mathbf X\right)\right)^{-\frac{1}{2}}\bS_{12}\left(\bX\right)\left(\mathbf S_{22}\left(\mathbf X\right)\right)^{-\frac{1}{2}}$ and so \\$\bS_{12}\left(\bX\right) = \left(\mathbf S_{11}\left(\mathbf X\right)\right)^{\frac{1}{2}} \hat{\mathbf{A}}\left(\mathbf X\right) \hat{\bm{\Lambda}}\left(\bX\right)\left(\hat{\mathbf{\Gamma}}\left(\mathbf X\right)\right)^\top\left(\mathbf S_{22}\left(\mathbf X\right)\right)^{\frac{1}{2}}$. In \eqref{eqn:selective_pvalue_definition}, we condition on the additional information
\begin{equation}
\label{eqn:conditioning_set_new}
    \left\{ \bS_{11} \left(\bX\right) = \bS_{11} \left(\bx\right), \bS_{22} \left(\bX\right) = \bS_{22} \left(\bx\right), \hat{\bA}\left(\bX\right) = \hat{\bA}\left(\bx\right), \hat{\bm{\Gamma}}\left(\bX\right) = \hat{\bm{\Gamma}}\left(\bx\right)\right\}.
\end{equation}
On this conditioning set, we have that  $\mathbf S_{12}\left(\mathbf X\right) = \left(\mathbf S_{11}\left(\mathbf x\right)\right)^{\frac{1}{2}} \hat{\mathbf{A}}\left(\mathbf x\right) \hat{\bm{\Lambda}}\left(\bX\right) \left(\hat{\mathbf{\Gamma}}\left(\mathbf x\right)\right)^\top \left(\mathbf S_{22}\left(\mathbf x\right)\right)^{\frac{1}{2}}$. Recalling the definition in \eqref{eqn:new_S}, we see that $\mathbf S_{12}\left(\mathbf X\right)  = \bS_{12}'\left(\mathbf x, \mathcal{R}, \left\{\hat\lambda_i \left(\bX\right)\right\}_{i= 1}^{p_1}\right)$. Thus, on the conditioning set \eqref{eqn:conditioning_set_new}, we have that
$$
\begin{aligned}
\bS\left(\bX\right) &= \begin{pmatrix}
\mathbf S_{11}\left(\mathbf x\right) & \mathbf S_{12}'\left(\mathbf x, \mathcal{R}, \left\{\hat\lambda_i \left(\bX\right)\right\}_{i= 1}^{p_1}\right) \\
\left( \mathbf S_{12}'\left(\mathbf x, \mathcal{R}, \left\{\hat\lambda_i \left(\bX\right)\right\}_{i= 1}^{p_1}\right) \right)^\top & \mathbf S_{22}\left(\mathbf x\right)
\end{pmatrix} =  \mathbf S'\left(\mathbf x, \mathcal{R}, \left\{\hat\lambda_i \left(\bX\right)\right\}_{i= 1}^{p_1}\right),
\end{aligned}
$$
where $\mathbf S'\left(\mathbf x, \mathcal{R}, \left\{\hat\lambda_i \left(\bX\right)\right\}_{i= 1}^{p_1}\right)$ was defined in Theorem \ref{thm:main}. Therefore, we can rewrite $p\left(\mathbf x; \mathcal{R}\right)$ in \eqref{eqn:selective_pvalue_definition}  as
$$
\begin{aligned}
    p\left(\mathbf x; \mathcal{R}\right) = \mathbb{P}_{H_0} \bigg(\prod_{i = 1}^{p_1} \left(1 - \left(\hat\lambda_i\left(\bX\right)\right)^2\right) \leqslant \prod_{i = 1}^{p_1} \left(1 - \left(\hat\lambda_i\left(\bx\right)\right)^2\right) \bigg |&
    \mathcal{R} \in \mathcal{P}\left(\bS'\left(\bx, \mathcal{R},  \left\{\lambda_i\left(\bX\right)\right\}_{i= 1}^{p_1}\right)\right),\\ 
&\bS_{11} \left(\bX\right) = \bS_{11} \left(\bx\right), \bS_{22} \left(\bX\right) = \bS_{22} \left(\bx\right),\\ &\hat{\bA}\left(\bX\right) = \hat{\bA}\left(\bx\right), \hat{\bm{\Gamma}}\left(\bX\right) = \hat{\bm{\Gamma}}\left(\bx\right)\bigg).\\
\end{aligned}
$$

From Lemma \ref{lemma:independence_under_null} we have that under $H_0$, $\left\{\hat\lambda_i\left(\bX\right)\right\}_{i= 1}^{p_1}$ is independent of \\$\left(\bS_{11}\left(\bX\right), \bS_{22}\left(\bX\right), \hat{\bA}\left(\bX\right), \hat{\bm{\Gamma}}\left(\bX\right)\right)$. It follows that
$$
\begin{aligned}
   p\left(\mathbf x; \mathcal{R}\right) = \mathbb{P}_{H_0} \bigg(\prod_{i = 1}^{p_1} \left(1 - \left(\hat\lambda_i\left(\bX\right)\right)^2\right) \leqslant \prod_{i = 1}^{p_1} \left(1 - \left(\hat\lambda_i\left(\bx\right)\right)^2\right) \bigg |&
    \mathcal{R} \in \mathcal{P}\left(\bS'\left(\bx, \mathcal{R},  \left\{\hat\lambda_i\left(\bX\right)\right\}_{i= 1}^{p_1}\right)\right)\bigg).
\end{aligned}
$$
Recall from Proposition \ref{prop:density_singular_values} that under $H_0$, $\left(\hat\lambda_1\left(\bX\right), \hat\lambda_2\left(\bX\right), \ldots, \hat\lambda_{p_1}\left(\bX\right)\right)$ has the density given in \eqref{eqn:density_singular_values}. Therefore, for $\left( \lambda_1, \lambda_2, \ldots, \lambda_{r\left( \mathcal P\right)}\right)$ with density given in \eqref{eqn:density_singular_values}, the first statement in Theorem \ref{thm:main} is established. 

The proof of \eqref{eqn:thm_type_1} is similar to the proof of Equation 11 in Theorem 1 of \citet{gao2020selective}. Using the probability integral transform we have that $p\left(\mathbf x; \mathcal{R}\right)$ in \eqref{eqn:selective_pvalue_definition} controls the selective type I error, i.e.
$$
\begin{aligned}
 \mathbb{P}_{H_0} \bigg(p\left(\mathbf X; \mathcal{R}\right) \leqslant \alpha \bigg | \mathcal{R} \in \mathcal{P}\left(\bS\left(\bX\right)\right)&, \bS_{11} \left(\bX\right) = \bS_{11} \left(\bx\right), \bS_{22} \left(\bX\right) = \bS_{22} \left(\bx\right), \\ &\:\hat{\bA}\left(\bX\right) = \hat{\bA}\left(\bx\right), \hat{\bm{\Gamma}}\left(\bX\right) = \hat{\bm{\Gamma}}\left(\bx\right)\bigg) \leqslant \alpha.
\end{aligned}
$$
By the law of iterated expectation, 
$$
\begin{aligned}
 &\mathbb{P}_{H_0} \left(p\left(\mathbf X; {\mathcal{R}}\right) \leqslant \alpha \bigg |{\mathcal{R}} \in \mathcal{P}\left(\bS\left(\bX\right)\right)\right)\\
 &=\mathbb{E}_{H_0} \left(\mathds{1}\left\{p\left(\mathbf X; {\mathcal{R}}\right) \leqslant \alpha \right\}\bigg |{\mathcal{R}} \in \mathcal{P}\left(\bS\left(\bX\right)\right)\right)\\
  &=\mathbb{E}_{H_0} \bigg( \mathbb{E}_{H_0} \bigg[\mathds{1}\left\{p\left(\mathbf X; {\mathcal{R}}\right) \leqslant \alpha\right\} \bigg |{\mathcal{R}} \in \mathcal{P}\left(\bS\left(\bX\right)\right), \bS_{11} \left(\bX\right) = \bS_{11} \left(\bx\right), \bS_{22} \left(\bX\right) = \bS_{22} \left(\bx\right),\\ &\:\:\:\:\:\:\:\:\:\:\:\:\:\:\:\:\:\:\:\:\:\:\:\:\:\:\:\:\:\:\:\:\:\:\:\:\:\:\:\:\:\:\:\:\:\:\:\:\:\:\:\:\:\:\:\:\:\:\:\:\:\:\: \hat{\bA}\left(\bX\right) = \hat{\bA}\left(\bx\right), \hat{\bm{\Gamma}}\left(\bX\right) = \hat{\bm{\Gamma}}\left(\bx\right)\bigg]\bigg |{\mathcal{R}} \in \mathcal{P}\left(\bS\left(\bX\right)\right)\bigg)\\
  &\leqslant \mathbb{E}_{H_0} \left( \alpha \bigg|{\mathcal{R}} \in \mathcal{P}\left(\bS\left(\bX\right)\right)\right)\\
  &= \alpha.
\end{aligned}
$$

This proves \eqref{eqn:thm_type_1} and completes the proof of Theorem \ref{thm:main}.

\subsection{Proof of Lemma \ref{lemma:independence_under_null}}
\label{Appendix:independence_under_null_lemma}
\begin{proof} 
In what follows, we omit the dependence of $\hat{\bm{\Lambda}}\left(\mathbf X\right), \mathbf{S}_{11}\left(\mathbf{X}\right), \mathbf{S}_{22}\left(\mathbf{X}\right), \hat{\mathbf{A}}\left(\mathbf{X}\right), \hat{\bm{\Gamma}}\left(\mathbf{X}\right)$ on $\bX$. Recall that \eqref{eqn:CCA} can be written as $\whitenS_{12} = \bS_{11}^{-\frac{1}{2}} \bS_{12} \bS_{22}^{-\frac{1}{2}} \overset{SVD}{=}\hat{\bA}\hat{\bm\Lambda} \hat{\bm\Gamma}^\top $. We first prove a technical lemma that will help us in establishing Lemma \ref{lemma:independence_under_null} going forward. 

\begin{lemma}
\label{lemma:conditional_indep}
Under $H_0: {\bm\Lambda} = \mathbf 0$, $\hat{\bm\Lambda}$ and $\left(\mathbf S_{11}, \hat{\mathbf A} , \hat{\bm\Gamma} \right)$ are conditionally independent  given $\mathbf{X}_2$, i.e. $\hat{\bm\Lambda} \independent \left(\mathbf S_{11}, \hat{\mathbf A},  \hat{\bm\Gamma} \right) \big | \mathbf{X}_2$. Furthermore, $ \left(\hat{\bm\Lambda}, \mathbf S_{11}, \hat{\mathbf A},  \hat{\bm\Gamma} \right) \independent \mathbf{X}_2$.
\end{lemma}

\begin{proof}
Let us define $\tilde{\mathbf W} := n\bS_{12}\bS_{22}^{-1}\bS_{21} = \tilde{\mathbf G}\tilde{\mathbf G}^\top,\: \tilde{\mathbf G} := \left( \tilde{\mathbf g}_1, \tilde{\mathbf g}_2, \ldots, \tilde{\mathbf g}_{p_2}\right), $ and $\tilde{\mathbf T} := n\mathbf S_{11} - n\bS_{12}\bS_{22}^{-1}\bS_{21} $. In order to obtain the  conditional joint density of $\tilde{\mathbf G} $ and $\tilde{\mathbf T} $ given $\mathbf X_2$, we first prove the following holds under $H_0: {\bm\Lambda} = \mathbf 0$:
\begin{enumerate}[1)]
    \item $\tilde{\mathbf g}_i\big |\mathbf X_2 \overset{i.i.d.}{\sim} N\left( \mathbf 0, \bSigma_{11}\right), i = 1,2,\ldots, p_2 $.
    \item $\tilde{\mathbf T} \big | \mathbf X_2 \sim \mathrm{Wishart}_{p_1} \left( \bSigma_{11}, n - 1 - p_2\right)$.
    \item $\tilde{\mathbf G} \independent \tilde{\mathbf T} \big | \mathbf X_2$.
\end{enumerate}

Let $\bH_n$ denote the centering matrix of size $n$, i.e. $\bH_n = \mathbf I_n - \frac{1}{n} \mathbf 1_n \mathbf 1_n^\top$. By construction, $\bH_n $ is symmetric and idempotent. We have
$$ 
\begin{aligned}
          \tilde{\mathbf W} &= \mathbf{X}_1^\top\bH_n \bX_2\left(\mathbf{X}_2^\top\bH_n \bX_2\right)^{-1}\mathbf{X}_2^\top\bH_n \bX_1 := \mathbf{X}_1^\top \mathbf Q_1 \mathbf{X}_1,\: \bQ_1 := \bH_n \bX_2\left(\mathbf{X}_2^\top\bH_n \bX_2\right)^{-1}\mathbf{X}_2^\top\bH_n,\\
          \tilde{\mathbf T} &= \mathbf{X}_1^\top\bH_n \bX_1 - \mathbf{X}_1^\top\bH_n \bX_2\left(\mathbf{X}_2^\top\bH_n \bX_2\right)^{-1}\mathbf{X}_2^\top\bH_n \bX_1 = \mathbf{X}_1^\top \left(\bH_n - \mathbf Q_1 \right)\mathbf{X}_1.
\end{aligned}
$$
We proceed as follows:
\begin{itemize}
    \item Given $\mathbf X_2$, under $H_0: {\bm\Lambda} = \mathbf 0$, we have that $\mathbf{X}_1 \sim \mathcal{MN}_{n \times p_1}(\mathbf 0, \mathbf I_n,  \bSigma_{11})$. We notice that $\bQ_1$ is a projection matrix onto the column space of $\bH_n\bX_2$, and therefore is idempotent with rank $p_2$. This implies that the eigenvalues of $\bQ_1$ are $1$ with multiplicity $p_2$ and $0$ with multiplicity $n - p_2$. Following the proof of Cochran's Theorem in Theorem 3.4.4 (a) of \citet{kent1979multivariate}, we have that $\mathbf{X}_1^\top \mathbf Q_1 \mathbf{X}_1 = \sum_{i = 1}^{p_2} \tilde{\mathbf g}_i \tilde{\mathbf g}_i ^\top$, with $\tilde{\mathbf g}_i = \mathbf X_1^\top \mathbf v_{(i)}$ where $\mathbf v_{(i)}$ is the $i^{th}$ eigenvector of $\mathbf Q_1$. Using Theorem 3.3.2 and 3.3.3 of \citet{kent1979multivariate}, we have that $\tilde{\mathbf g}_i\big |\mathbf X_2 \overset{i.i.d.}{\sim} N\left( \mathbf 0, \bSigma_{11}\right), i = 1,2,\ldots, p_2 $. This proves 1). 
    
    \item Next, we observe that $\bH_n - \bQ_1$ is an idempotent matrix, with $\mathrm{rank} \left( \bH_n - \bQ_1\right) = \mathrm{trace}(\bH_n) - \mathrm{trace}(\bQ_1 ) = n  - 1 - p_2$. Theorem 3.4.4 (b) of \citet{kent1979multivariate} implies $\tilde{\mathbf T} \big | \mathbf X_2 \sim \mathrm{Wishart}_{p_1} \left( \bSigma_{11}, n - 1 - p_2\right)$.  This proves 2). 
    
    \item We notice that $\left( \bH_n - \bQ_1\right)\bQ_1 = \bH_n \bQ_1 - \bQ_1 = \mathbf 0$, where the first and the last equality follows from the idempotence of $\bQ_1 $ and  $\bH_n$, respectively. Using Craig’s Theorem in Theorem 3.4.5 of \citet{kent1979multivariate}, this implies $\tilde{\mathbf W} \independent \tilde{\mathbf T} \big | \mathbf X_2$. As shown in the proof of 1), given $\mathbf X_2$, $\tilde{\mathbf G}$ is a function of $\tilde{\mathbf{W}}$. Hence, $\tilde{\mathbf G} \independent \tilde{\mathbf T} \big | \mathbf X_2$. This proves 3). 
\end{itemize}

Combining 1)-3) above, we obtain the conditional joint density of $\left(\tilde{\mathbf G}, \tilde{\mathbf T}\right)$, given $\mathbf X_2$:
$$
\begin{aligned}
          f\left( \tilde{\mathbf G} = \left( \tilde{\mathbf g}_1, \tilde{\mathbf g}_2, \ldots, \tilde{\mathbf g}_{p_2} \right), \tilde{\mathbf T} \big| \mathbf X_2\right) &\propto \left|\tilde{\mathbf T}\right|^{\frac{n - p - 1}{2}}\exp\left[ -\frac{1}{2}\left(\mathrm{tr}\left( \bSigma^{-1}_{11}\tilde{\mathbf T}\right) + \sum_{i = 1}^{p_2}\tilde{\mathbf g}_i ^\top\bSigma^{-1}_{11}\tilde{\mathbf g}_i  \right) \right]\\
          &= \left| \tilde{\mathbf T}\right|^{\frac{n - p - 1}{2}}\exp\left[ -\frac{1}{2}\left(\mathrm{tr}\left( \bSigma^{-1}_{11}\tilde{\mathbf T}\right) + \sum_{i = 1}^{p_2}\mathrm{tr}\left(\mathbf \bSigma^{-1}_{11}\tilde{\mathbf g}_i\tilde{\mathbf g}_i ^\top \right) \right) \right]\\
          &= \left| \tilde{\mathbf T}\right|^{\frac{n - p - 1}{2}}\exp\left[ -\frac{1}{2}\left(\mathrm{tr}\left( \bSigma^{-1}_{11}\left( \tilde{\mathbf T}+\tilde{\mathbf G}\tilde{\mathbf G}^\top\right)\right) \right) \right].
\end{aligned}
$$
Next, motivated by \citet{hsu1939distribution}, we consider the change of variable $\left(\tilde{\mathbf G}, \tilde{\mathbf T} \right) \mapsto \left( \mathbf S_{11}, \hat{\mathbf A},  \hat{\bm\Lambda} , \hat{\bm\Gamma} \right)$. Here, $n\bS_{11} = \tilde{\mathbf T} + \tilde{\mathbf G}\tilde{\mathbf G}^\top$ and $\hat{\mathbf A}\hat{\bm\Lambda}\hat{\bm\Gamma} ^\top = \whitenS_{12} =  n^{-\frac{1}{2}}\bS_{11}^{- \frac{1}{2}}\tilde{\mathbf G}$. From (2)-(4) in \citet{ben1999change}, we have 
\begin{equation}
\label{eqn:variable_transform}
    \begin{aligned}
          f\left( \mathbf S_{11}, \hat{\mathbf A},  \hat{\bm\Lambda} , \hat{\bm\Gamma} \big| \mathbf X_2 \right) \propto& f\left(\tilde{\mathbf G} , \tilde{\mathbf T} \big| \mathbf X_2\right) \times\\
          &\left| \mathbf J_{\left(\tilde{\mathbf G}, \tilde{\mathbf T} \right) \mapsto \left( \mathbf S_{11}, \hat{\mathbf A},  \hat{\bm\Lambda} , \hat{\bm\Gamma} \right)} \left(\mathbf J_{\left(\tilde{\mathbf G}, \tilde{\mathbf T} \right) \mapsto \left( \mathbf S_{11}, \hat{\mathbf A},  \hat{\bm\Lambda} , \hat{\bm\Gamma} \right)}\right)^\top\right|^{-\frac{1}{2}},
\end{aligned}
\end{equation}
where $\mathbf J_{\left(\tilde{\mathbf G}, \tilde{\mathbf T} \right) \mapsto \left( \mathbf S_{11}, \hat{\mathbf A},  \hat{\bm\Lambda} , \hat{\bm\Gamma} \right)}$ is the Jacobian matrix of the transformation $\left(\tilde{\mathbf G}, \tilde{\mathbf T} \right) \mapsto \left( \mathbf S_{11}, \hat{\mathbf A},  \hat{\bm\Lambda} , \hat{\bm\Gamma} \right)$. Plugging in $ \tilde{\mathbf G} = n^{\frac{1}{2}}\bS_{11}^{\frac{1}{2}}\hat{\mathbf A}\hat{\bm\Lambda}\hat{\bm\Gamma} ^\top$ and $\tilde{\mathbf T} =  n\bS_{11} - \tilde{\mathbf G}\tilde{\mathbf G}^\top$ into  $f\left(\tilde{\mathbf G} , \tilde{\mathbf T} \big| \mathbf X_2\right)$ shows that it factors across the variables $\bS_{11},\hat{\mathbf A},\hat{\bm\Lambda}$ and $ \hat{\bm\Gamma}$:
\begin{equation}
    \label{eqn:simplification}
    \begin{aligned}
f\left(\tilde{\mathbf G} , \tilde{\mathbf T} \big| \mathbf X_2\right)
          &\propto \left| n\mathbf S_{11} - n\bS_{11}^{ \frac{1}{2}}\hat{\mathbf A}\hat{\bm\Lambda}\hat{\bm\Gamma}^\top \hat{\bm\Gamma}\hat{\bm\Lambda}\hat{\mathbf A}^\top  \bS_{11}^{ \frac{1}{2}}  \right|^{\frac{n - p - 1}{2}}\exp\left[ -\frac{1}{2}\left(\mathrm{tr}\left( n\bSigma_{11}^{-1} \bS_{11}\right) \right) \right]\\
          &\propto \left| \mathbf I_{p_1} - \hat{\mathbf A}\hat{\bm\Lambda}^2\hat{\mathbf A}^\top \right|^{\frac{n - p - 1}{2}}\left| \mathbf S_{11}\right|^{-\frac{n - p - 1 }{2}}\exp\left[ -\frac{1}{2}\left(\mathrm{tr}\left( n\bSigma_{11}^{-1} \bS_{11}\right) \right) \right]\\
           &=\left| \mathbf I_{p_1} - \hat{\bm\Lambda}^2 \right|^{\frac{n - p - 1}{2}}\left| \mathbf S_{11}\right|^{-\frac{n - p - 1 }{2}}\exp\left[ -\frac{1}{2}\left(\mathrm{tr}\left( n\bSigma_{11}^{-1} \bS_{11}\right) \right) \right],
\end{aligned}
\end{equation}
where the last equality follows by observing $\left| \mathbf I - \mathbf Z \mathbf Z^\top\right|  = \left| \mathbf I - \mathbf Z^\top \mathbf Z\right|$, with the square matrix $\mathbf Z = \hat{\mathbf A}\hat{\bm\Lambda}$. Thus, it remains to show that the Jacobian term in \eqref{eqn:variable_transform} also factors. Now,

$$
\begin{aligned}
           \mathbf \mathbf J_{\left(\tilde{\mathbf G}, \tilde{\mathbf T} \right) \mapsto \left( \mathbf S_{11}, \hat{\mathbf A},  \hat{\bm\Lambda} , \hat{\bm\Gamma} \right)} &= \begin{pmatrix}
\frac{\partial \bS_{11}}{\partial \tilde{\mathbf G}} & \frac{\partial \hat{\mathbf V}}{\partial \tilde{\mathbf G}}\\
\frac{\partial \bS_{11}}{\partial \tilde{\mathbf T}} & \frac{\partial \hat{\mathbf V}}{\partial \tilde{\mathbf T}}
\end{pmatrix},
\end{aligned}
$$
where we write $\hat{\mathbf V} = \left(\hat{\mathbf A},  \hat{\bm\Lambda} , \hat{\bm\Gamma}\right)$. Noting that $\frac{\partial \hat{\mathbf V}}{\partial \tilde{\mathbf T}} = \frac{\partial {\whitenS_{12}}}{\partial \tilde{\mathbf T}}\frac{\partial \hat{\mathbf V}}{\partial{\whitenS_{12}}} = \mathbf 0$ and $\frac{\partial \bS_{11}}{\partial \tilde{\mathbf T}} = n^{-1}\mathbf I$,
$$
\begin{aligned}
          \mathbf J_{\left(\tilde{\mathbf G}, \tilde{\mathbf T} \right) \mapsto \left( \mathbf S_{11}, \hat{\mathbf A},  \hat{\bm\Lambda} , \hat{\bm\Gamma} \right)} \left(\mathbf J_{\left(\tilde{\mathbf G}, \tilde{\mathbf T} \right) \mapsto \left( \mathbf S_{11}, \hat{\mathbf A},  \hat{\bm\Lambda} , \hat{\bm\Gamma} \right)}\right)^\top  = \begin{pmatrix}\frac{\partial \bS_{11}}{\partial \tilde{\mathbf G}}\left( \frac{\partial \bS_{11}}{\partial \tilde{\mathbf G}}\right)^\top + \frac{\partial \hat{\mathbf V}}{\partial \tilde{\mathbf G}}\left( \frac{\partial \hat{\mathbf V}}{\partial \tilde{\mathbf G}}\right)^\top   & n^{-1}\frac{\partial \bS_{11}}{\partial \tilde{\mathbf G}}\\ n^{-1}\left( \frac{\partial \bS_{11}}{\partial \tilde{\mathbf G}}\right)^\top & n^{-2}\mathbf I
\end{pmatrix}.
\end{aligned}
$$
The determinant of a block matrix $\begin{pmatrix}\mathbf A & \mathbf B \\
\mathbf C & \mathbf D\end{pmatrix}$, with an invertible $\mathbf D$, is given by $\left| \mathbf A - \mathbf B \mathbf D^{-1} \mathbf C\right|\left| \mathbf D\right|$. Applying this to $\mathbf J_{\left(\tilde{\mathbf G}, \tilde{\mathbf T} \right) \mapsto \left( \mathbf S_{11}, \hat{\mathbf A},  \hat{\bm\Lambda} , \hat{\bm\Gamma} \right)} \left(\mathbf J_{\left(\tilde{\mathbf G}, \tilde{\mathbf T} \right) \mapsto \left( \mathbf S_{11}, \hat{\mathbf A},  \hat{\bm\Lambda} , \hat{\bm\Gamma} \right)}\right)^\top$, we have 
\begin{equation}
\label{eqn:variable_transform_jacobian}
    \begin{aligned}
         & \left| \mathbf J_{\left(\tilde{\mathbf G}, \tilde{\mathbf T} \right) \mapsto \left( \mathbf S_{11}, \hat{\mathbf A},  \hat{\bm\Lambda} , \hat{\bm\Gamma} \right)} \left(\mathbf J_{\left(\tilde{\mathbf G}, \tilde{\mathbf T} \right) \mapsto \left( \mathbf S_{11}, \hat{\mathbf A},  \hat{\bm\Lambda} , \hat{\bm\Gamma} \right)}\right)^\top\right| \\
          &= \left| \frac{\partial \hat{\mathbf V}}{\partial \tilde{\mathbf G}}\left( \frac{\partial \hat{\mathbf V}}{\partial \tilde{\mathbf G}}\right)^\top\right| = \left|  \frac{\partial \hat{\whitenS_{12}}}{\partial \tilde{\mathbf G}}\frac{\partial \hat{\mathbf V}}{\partial\hat{\whitenS_{12}}}\left( \frac{\partial \hat{\whitenS_{12}}}{\partial \tilde{\mathbf G}}\frac{\partial \hat{\mathbf V}}{\partial\hat{\whitenS_{12}}}\right)^\top\right|\\
          &= \left| \mathbf S_{11}^{-\frac{1}{2}}\right|^{p_2}\left|\left(\mathbf J_{\left(\whitenS_{12} \right) \mapsto \left( \hat{\mathbf A},  \hat{\bm\Lambda} , \hat{\bm\Gamma} \right)}\right)\left(\mathbf J_{\left(\whitenS_{12} \right) \mapsto \left( \hat{\mathbf A},  \hat{\bm\Lambda} , \hat{\bm\Gamma} \right)}\right)^\top \right|\left| \mathbf S_{11}^{-\frac{1}{2}}\right|^{p_2}\\
          &=\left| \mathbf S_{11}\right|^{-p_2}\left| \mathbf J_{\left(\whitenS_{12} \right) \mapsto \left( \hat{\mathbf A},  \hat{\bm\Lambda} , \hat{\bm\Gamma} \right)} \left(\mathbf J_{\left(\whitenS_{12} \right) \mapsto \left( \hat{\mathbf A},  \hat{\bm\Lambda} , \hat{\bm\Gamma} \right)}\right)^\top\right|,
\end{aligned}
\end{equation}
where we have used the chain rule for matrices, and properties of the determinant, and the fact that $\frac{\partial \hat{\mathbf V}}{\partial\hat{\whitenS_{12}}} = \mathbf J_{\left(\whitenS_{12} \right) \mapsto \left( \hat{\mathbf A},  \hat{\bm\Lambda} , \hat{\bm\Gamma} \right)}$. In (10) of \citet{rennie2006jacobian}, the Jacobian of the transformation $\left( \whitenS_{12}\right) \mapsto \left( \hat{\mathbf A},  \hat{\bm\Lambda} , \hat{\bm\Gamma} \right)$ was shown to be a separable function of $\left(\hat{\bA},  \hat{\bm\Lambda} , \hat{\bm\Gamma}\right)$ not dependent on $\mathbf X_2$. This implies that
$f\left( \mathbf S_{11}, \hat{\mathbf A},  \hat{\bm\Lambda} , \hat{\bm\Gamma} \big| \mathbf X_2 \right) \propto f\left( \mathbf S_{11} \big| \mathbf X_2 \right)  f\left( \hat{\mathbf A} \big| \mathbf X_2 \right)  f\left(  \hat{\bm\Lambda} \big| \mathbf X_2 \right)  f\left(  \hat{\bm\Gamma} \big| \mathbf X_2 \right) $, proving the conditional independence result. Lastly, observe that $\mathbf X_2$ does not appear anywhere in the conditional density, and thus $\left( \mathbf S_{11}, \hat{\mathbf A},  \hat{\bm\Lambda} , \hat{\bm\Gamma} \right) \independent \mathbf X_2$.
\end{proof}

\noindent From Lemma \ref{lemma:conditional_indep}, it follows that $\hat{\bm\Lambda} \independent \left(\mathbf S_{11}, \hat{\mathbf A},  \hat{\bm\Gamma} \right) $ and $\hat{\bm\Lambda} \independent \mathbf X_2$, which completes the proof of Lemma \ref{lemma:independence_under_null}. 
\end{proof}

\section{Proof of Proposition \ref{prop:reformulate_conditioning_set}}
\label{Appendix:conditioning_set}
First, we assume that the variables are unordered. For $\hat{\mathcal{P}}\in \mathcal{P}\left( \mathbf S \left( \mathbf x\right)\right)$, we want to show that
\begin{equation}\label{eqn:prop_2}
    \hat{\mathcal{P}} \in \mathcal{P}\left( \mathbf S'\left(\mathbf x, \hat{\mathcal{P}}, \left\{\lambda_i\right\}_{i= 1}^{r\left(\hat{\mathcal{P}}\right)}\right)\right) \iff \max_{i' \in \hat{\mathcal P}, j' \in \hat{\mathcal P}^c}\left|\mathbf R_{i',j'}\left(\mathbf S'\left(\mathbf x, \hat{\mathcal{P}}, \left\{\lambda_i\right\}_{i= 1}^{r\left(\hat{\mathcal{P}}\right)}\right)\right)\right|\leqslant c.
\end{equation}

First we prove $\hat{\mathcal{P}} \in \mathcal{P}\left( \mathbf S'\left(\mathbf x, \hat{\mathcal{P}}, \left\{\lambda_i\right\}_{i= 1}^{r\left(\hat{\mathcal{P}}\right)}\right)\right) \implies \max_{i' \in \hat{\mathcal P}, j' \in \hat{\mathcal P}^c}\left|\mathbf R_{i',j'}\left(\mathbf S'\left(\mathbf x, \hat{\mathcal{P}}, \left\{\lambda_i\right\}_{i= 1}^{r\left(\hat{\mathcal{P}}\right)}\right)\right)\right|\leqslant c$. Following Algorithm \ref{algo:Selection_procedure}, in order to compute $\mathcal{P}\left( \mathbf S'\left(\mathbf x, \hat{\mathcal{P}}, \left\{\lambda_i\right\}_{i= 1}^{r\left(\hat{\mathcal{P}}\right)}\right)\right)$, we consider the adjacency matrix $\bD\left( \mathbf S'\left(\mathbf x, \hat{\mathcal{P}}, \left\{\lambda_i\right\}_{i= 1}^{r\left(\hat{\mathcal{P}}\right)}\right)\right)$, where for $i', j' \in \{1,2,\ldots, p \}$,
$$
\bD_{i',j'}\left( \mathbf S'\left(\mathbf x, \hat{\mathcal{P}}, \left\{\lambda_i\right\}_{i= 1}^{r\left(\hat{\mathcal{P}}\right)}\right)\right) := \mathds{1}\left\{ \left|\bR_{i',j'}\left( \mathbf S'\left(\mathbf x, \hat{\mathcal{P}}, \left\{\lambda_i\right\}_{i= 1}^{r\left(\hat{\mathcal{P}}\right)}\right)\right)\right| > c\right\}.
$$
If $\hat{\mathcal{P}}\in \mathcal P\left(\mathbf S'\left(\mathbf x,\hat{\mathcal{P}}, \left\{\lambda_i\right\}_{i= 1}^{r\left(\hat{\mathcal{P}}\right)}\right)\right)$, then in the undirected graph corresponding to \\ $\bD\left( \mathbf S'\left(\mathbf x, \hat{\mathcal{P}}, \left\{\lambda_i\right\}_{i= 1}^{r\left(\hat{\mathcal{P}}\right)}\right)\right)$, the nodes in $\hat{\mathcal{P}}$ are not connected with any of the nodes in $\hat{\mathcal P}^c$. Hence,
\begin{equation}
    \label{eqn:equivalence}
    \begin{aligned}
    &\hat{\mathcal{P}}\in \mathcal P\left(\mathbf S'\left(\mathbf x,\hat{\mathcal{P}}, \left\{\lambda_i\right\}_{i= 1}^{r\left(\hat{\mathcal{P}}\right)}\right)\right)\\
          &\implies\bD_{i', j'}\left( \mathbf S'\left(\mathbf x, \hat{\mathcal{P}}, \left\{\lambda_i\right\}_{i= 1}^{r\left(\hat{\mathcal{P}}\right)}\right)\right) = 0, \text{ for all } i' \in \hat{\mathcal P}; \: j' \in \hat{\mathcal P}^c\\
          &\implies \left|\bR_{i', j'}\left( \mathbf S'\left(\mathbf x, \hat{\mathcal{P}}, \left\{\lambda_i\right\}_{i= 1}^{r\left(\hat{\mathcal{P}}\right)}\right)\right)\right| \leqslant c,  \text{ for all }  i' \in \hat{\mathcal P}; \: j' \in \hat{\mathcal P}^c\\
          &\implies  \max_{i' \in \hat{\mathcal P}, j' \in \hat{\mathcal P}^c}\left|\mathbf R_{i',j'}\left(\mathbf S'\left(\mathbf x, \hat{\mathcal{P}}, \left\{\lambda_i\right\}_{i= 1}^{r\left(\hat{\mathcal{P}}\right)}\right)\right)\right|\leqslant c.
\end{aligned}
\end{equation} 
This completes the proof that $\mathrm{LHS} \implies \mathrm{RHS}$ in \eqref{eqn:prop_2}. 

Next, we prove $\max_{i' \in \hat{\mathcal P}, j' \in \hat{\mathcal P}^c}\left|\mathbf R_{i',j'}\left(\mathbf S'\left(\mathbf x, \hat{\mathcal{P}}, \left\{\lambda_i\right\}_{i= 1}^{r\left(\hat{\mathcal{P}}\right)}\right)\right)\right|\leqslant c \implies $ \\$\hat{\mathcal{P}} \in \mathcal{P}\left( \mathbf S'\left(\mathbf x, \hat{\mathcal{P}}, \left\{\lambda_i\right\}_{i= 1}^{r\left(\hat{\mathcal{P}}\right)}\right)\right)$. Recall that by its definition in \eqref{eqn:new_S}, $\mathbf S'\left(\mathbf x, \hat{\mathcal{P}}, \left\{\lambda_i\right\}_{i= 1}^{r\left(\hat{\mathcal{P}}\right)}\right)$ is a modified version of $\bS\left(\bx\right)$ with perturbed off-diagonal blocks. This gives us
\begin{equation}
\label{eqn:equal_diag_blocks}
\mathbf R_{i',j'}\left(\mathbf S'\left(\mathbf x, \hat{\mathcal{P}}, \left\{\lambda_i\right\}_{i= 1}^{r\left(\hat{\mathcal{P}}\right)}\right)\right) = \mathbf R_{i',j'}\left(\mathbf S\left(\mathbf x\right)\right), \text{ for } i', j' \in \hat{\mathcal{P}} \text{ or } i', j' \in \hat{\mathcal{P}}^c.
\end{equation}
Moreover, by Algorithm \ref{algo:Selection_procedure}, $\hat{\mathcal{P}}\in \mathcal{P}\left( \mathbf S \left( \mathbf x\right)\right)$ implies $\max_{i' \in \hat{\mathcal P}, j' \in \hat{\mathcal P}^c}\left|\mathbf R_{i',j'}\left(\mathbf S\left(\mathbf x\right)\right)\right|\leqslant c$. Combining this with \eqref{eqn:equal_diag_blocks} implies that $\bD\left( \mathbf S'\left(\mathbf x, \hat{\mathcal{P}}, \left\{\lambda_i\right\}_{i= 1}^{r\left(\hat{\mathcal{P}}\right)}\right)\right) = \bD\left( \mathbf S\left(\mathbf x \right)\right)$. Recalling that $\mathcal{P}\left( \mathbf S'\left(\mathbf x,\hat{\mathcal{P}}, \left\{\lambda_i\right\}_{i= 1}^{r\left(\hat{\mathcal{P}}\right)}\right)\right)$ depends only on the adjacency matrix $\mathbf D \left( \mathbf S'\left(\mathbf x,\hat{\mathcal{P}}, \left\{\lambda_i\right\}_{i= 1}^{r\left(\hat{\mathcal{P}}\right)}\right) \right)$ (see Algorithm \ref{algo:Selection_procedure}), it follows that $\hat{\mathcal{P}}\in \mathcal P\left(\mathbf S'\left(\mathbf x,\hat{\mathcal{P}}, \left\{\lambda_i\right\}_{i= 1}^{r\left(\hat{\mathcal{P}}\right)}\right)\right)$. This completes the proof in the case of the unordered variables. 

In the case of ordered variables, to show that $\max_{i' \in \hat{\mathcal P}, j' \in \hat{\mathcal P}^c}\left|\mathbf R_{i',j'}\left(\mathbf S'\left(\mathbf x, \hat{\mathcal{P}}, \left\{\lambda_i\right\}_{i= 1}^{r\left(\hat{\mathcal{P}}\right)}\right)\right)\right|\leqslant c$ implies $  \hat{\mathcal{P}} \in \mathcal{P}\left( \mathbf S'\left(\mathbf x, \hat{\mathcal{P}}, \left\{\lambda_i\right\}_{i= 1}^{r\left(\hat{\mathcal{P}}\right)}\right)\right)$, we also need to show that $\hat{\mathcal P} $ consists only of consecutive variables. This follows immediately from the fact that $\hat{\mathcal{P}}\in \mathcal{P}\left( \mathbf S \left( \mathbf x\right)\right)$. This completes the proof in the case of ordered variables. 

\section{Proof of Proposition \ref{prop:linear1}}
\label{Appendix:prop_linear1}
From \eqref{eqn:new_S}, $\mathbf S'\left(\mathbf x, \hat{\mathcal{P}}, \left\{\lambda_i\right\}_{i= 1}^{r\left(\hat{\mathcal{P}}\right)}\right)$ is linear in $\left(\lambda_1, \lambda_2, \ldots, \lambda_{r\left(\hat{\mathcal{P}}\right)}\right)$. Furthermore, \\$\bH := \mathrm{diag}\left(\mathbf S'\left(\mathbf x, \hat{\mathcal{P}}, \left\{\lambda_i\right\}_{i= 1}^{r\left(\hat{\mathcal{P}}\right)}\right)\right)$ is not a function of $\left(\lambda_1, \lambda_2, \ldots, \lambda_{r\left(\hat{\mathcal{P}}\right)}\right)$, since only the off-diagonal blocks of $\mathbf S'\left(\mathbf x, \hat{\mathcal{P}}, \left\{\lambda_i\right\}_{i= 1}^{r\left(\hat{\mathcal{P}}\right)}\right)$ are functions of $\left(\lambda_1, \lambda_2, \ldots, \lambda_{r\left(\hat{\mathcal{P}}\right)}\right)$. Hence, $\bR\left(\mathbf S'\left(\mathbf x, \hat{\mathcal{P}}, \left\{\lambda_i\right\}_{i= 1}^{r\left(\hat{\mathcal{P}}\right)}\right)\right) = \bH^{-\frac{1}{2}}\mathbf S'\left(\mathbf x, \hat{\mathcal{P}}, \left\{\lambda_i\right\}_{i= 1}^{r\left(\hat{\mathcal{P}}\right)}\right)\bH^{-\frac{1}{2}}$ is linear in $\left(\lambda_1, \lambda_2, \ldots, \lambda_{r\left(\hat{\mathcal{P}}\right)}\right)$. Furthermore, $\mathbf{R}_{\hat{\mathcal{P}}, \hat{\mathcal{P}}^c}\left(\mathbf S'\left(\mathbf x, \hat{\mathcal{P}}, \left\{\lambda_i\right\}_{i= 1}^{r\left(\hat{\mathcal{P}}\right)}\right)\right)$ is a submatrix of $\mathbf{R}\left(\mathbf S'\left(\mathbf x, \hat{\mathcal{P}}, \left\{\lambda_i\right\}_{i= 1}^{r\left(\hat{\mathcal{P}}\right)}\right)\right)$, and hence is linear in $\left(\lambda_1, \lambda_2, \ldots, \lambda_{r\left(\hat{\mathcal{P}}\right)}\right)$.

\section{Proof of Theorem \ref{thm:character}}
\label{Appendix:character}
The conditioning set ${\mathcal{G}}\left(\mathbf x; \hat{\mathcal{P}} \right)$ in \eqref{eqn:conditioning_set} can be rewritten as in \eqref{eqn:polytope_proof}. Thus, 
$$
    \begin{aligned}
     &{\mathcal{G}}\left(\mathbf x; \hat{\mathcal{P}} \right) \\
     &= \Bigg\{ \left(\lambda_1, \lambda_2, \ldots, \lambda_{r\left(\hat{\mathcal{P}}\right)}\right) \in [0, 1]^{r\left(\hat{\mathcal{P}}\right)} :  \lambda_1 \geqslant \lambda_2 \geqslant \ldots \geqslant \lambda_{r\left(\hat{\mathcal{P}}\right)},\\ 
     &\:\:\:\:\:\:\:\:\:\:\:\:\:\:\:\:\:\:\:\:\:\:\:\:\:\:\:\:\:\:\:\:\:\:\:\:\:\:\:\:\:\:\:\:\:\:\:\:\:\:\:\:\begin{pmatrix}
\mathbf{R}_{\hat{\mathcal{P}}, \hat{\mathcal{P}}^c}\left(\mathbf S'\left(\mathbf x, \hat{\mathcal{P}}, \left\{\lambda_i\right\}_{i= 1}^{r\left(\hat{\mathcal{P}}\right)}\right)\right) \\
-\mathbf{R}_{\hat{\mathcal{P}}, \hat{\mathcal{P}}^c}\left(\mathbf S'\left(\mathbf x, \hat{\mathcal{P}}, \left\{\lambda_i\right\}_{i= 1}^{r\left(\hat{\mathcal{P}}\right)}\right)\right)
\end{pmatrix} \leqslant c \mathbf{1}_{2\mathrm{card}\left(\hat{\mathcal P}\right)} \mathbf 1^\top_{\left( p - \mathrm{card}\left(\hat{\mathcal P}\right)\right)}\Bigg\}.
    \end{aligned}
$$
The proof follows by combining the $2r\left(\hat{\mathcal{P}}\right)\left(p - r\left(\hat{\mathcal{P}}\right)\right)$ inequalities resulting from the $\max$ term in \eqref{eqn:polytope_proof} with the $r\left(\hat{\mathcal{P}}\right) + 1$ inequalities associated with upper and lower limits and ordering restrictions on $\left(\lambda_1, \lambda_2, \ldots, \lambda_{r\left(\hat{\mathcal{P}}\right)}\right)$, i.e. $\lambda_1 \leqslant 1, \lambda_{l} - \lambda_{l-1} \leqslant 0, $ for $ l = 2, \ldots, r\left(\hat{\mathcal{P}}\right), $ and $  - \lambda_{r\left(\hat{\mathcal{P}}\right)} \leqslant 0$.

\section{Computation of $p\left(\mathbf x; \hat{\mathcal{P}}\right)$ in \eqref{eqn:probability}}
\label{Appendix:importance_sampling}
Here, we elaborate on the discussion in Section \ref{sec:summary}. The subsections here follow the organization of Section \ref{sec:summary}. The overall procedure can be found in Algorithm \ref{algo:SI}. 

\subsection{Proof of Proposition \ref{prop:univariate_beta}}
\label{Appendix:singleton_beta}
Without loss of generality, we assume $\mathrm{card}\left(\hat{\mathcal{P}}\right) = p - 1$. 
Next, we characterize the conditioning set in \eqref{eqn:conditioning_set} for $\mathrm{card}\left(\hat{\mathcal{P}}\right) = p - 1$. From Theorem \ref{thm:main}, we observe that the conditioning set is a function of a modified version of the sample covariance matrix with a perturbed off-diagonal block. Hence we primarily focus on the off-diagonal block of the covariance matrix. We denote the vectors  $\whitenS_{\hat{\mathcal{P}}, \hat{\mathcal{P}}^c} \left(\bx\right)$ and  $\bS_{\hat{\mathcal{P}}, \hat{\mathcal{P}}^c} \left(\bx\right)$ as $\whitens_{\hat{\mathcal{P}}, \hat{\mathcal{P}}^c} \left(\bx\right)$ and $\bs_{\hat{\mathcal{P}}, \hat{\mathcal{P}}^c} \left(\bx\right)$, respectively. We denote the scalar $\bS_{\hat{\mathcal{P}}^c, \hat{\mathcal{P}}^c} \left(\bx\right)$ as  $s_{\hat{\mathcal{P}}^c, \hat{\mathcal{P}}^c} \left(\bx\right)$.  Using  \eqref{eqn:CCA} we have 
\begin{equation}
\label{eqn:SVD_r1}
    \whitens_{\hat{\mathcal{P}}, \hat{\mathcal{P}}^c} \left(\bx\right) = \frac{ \left(\bS_{{\hat{\mathcal{P}}}, {\hat{\mathcal{P}}}}\left(\bx\right)\right)^{-\frac{1}{2}}\bs_{\hat{\mathcal{P}}, \hat{\mathcal{P}}^c}\left(\bx\right)}{\sqrt{s_{\hat{\mathcal{P}}^c, \hat{\mathcal{P}}^c}\left(\bx\right)}} \overset{SVD}{=} \hat{\mathbf{A}}^{{\mathcal{P}}}\left(\bx\right) \: \hat{\bm\Lambda}^{{\mathcal{P}}}\left(\bx\right)\:  \left({\hat{\mathbf{\Gamma}}^{{{\mathcal{P}}}}}\left(\bx\right)\right)^\top.
\end{equation}

As $\hat{\mathbf{A}}^{{\mathcal{P}}}\left(\bx\right)  \hat{\bm\Lambda}^{{\mathcal{P}}}\left(\bx\right) \left({\hat{\mathbf{\Gamma}}^{{{\mathcal{P}}}}}\left(\bx\right)\right)^\top$ is the SVD decomposition of the vector $\whitens_{\hat{\mathcal{P}}, \hat{\mathcal{P}}^c} \left(\bx\right)$, the following holds true in  \eqref{eqn:SVD_r1}:
\begin{enumerate}
    \item $\hat{\bm\Lambda}^{{\mathcal{P}}}\left(\bx\right)$ is a scalar equal to $ \left\| \whitens_{\hat{\mathcal{P}}, \hat{\mathcal{P}}^c} \left(\bx\right)\right\|_2$. We denote this as $\hat\lambda_1 ^{\hat{\mathcal P}}\left(\mathbf x\right) $. 
    \item ${\hat{\mathbf{\Gamma}}^{{{\mathcal{P}}}}}\left(\bx\right) = 1$.
    \item $\hat{\mathbf{A}}^{{\mathcal{P}}}\left(\bx\right) = \frac{1}{\left\| \whitens_{\hat{\mathcal{P}}, \hat{\mathcal{P}}^c} \left(\bx\right)\right\|_2}\whitens_{\hat{\mathcal{P}}, \hat{\mathcal{P}}^c} \left(\bx\right)$. 
\end{enumerate}
Substituting these in \eqref{eqn:SVD_r1}, we have
\begin{equation}
    \label{eqn:SVD_r2}
    \bs_{\hat{\mathcal{P}}, \hat{\mathcal{P}}^c}\left( \bx \right) = \hat\lambda_1 ^{\hat{\mathcal P}}\left(\mathbf x\right) \sqrt{s_{\hat{\mathcal{P}}^c, \hat{\mathcal{P}}^c}\left(\bx\right)} \left(\bS_{{\hat{\mathcal{P}}}, {\hat{\mathcal{P}}}}\left(\bx\right)\right)^{\frac{1}{2}}\frac{1}{\left\| \whitens_{\hat{\mathcal{P}}, \hat{\mathcal{P}}^c} \left(\bx\right)\right\|_2}\whitens_{\hat{\mathcal{P}}, \hat{\mathcal{P}}^c} \left(\bx\right).
\end{equation}

Next for $1 \leqslant \lambda_1 \leqslant  0$, we focus on the perturbed off-diagonal block of the modified covariance matrix $\mathbf S'\left(\mathbf x, \hat{\mathcal{P}}, \left\{\lambda_1\right\} \right)$ in \eqref{eqn:conditioning_set}. For $\mathrm{card}\left(\hat{\mathcal{P}}^c\right) = 1$, $\mathbf S_{\hat{\mathcal{P}}, \hat{\mathcal{P}}^c}'$ in \eqref{eqn:new_S} can be rewritten as:
\begin{equation}
\label{eqn:rp=1}
    \begin{aligned}
              \mathbf S_{\hat{\mathcal{P}}, \hat{\mathcal{P}}^c}'\left(\mathbf x, \hat{\mathcal{P}}, \left\{\lambda_1\right\}\right) &= \left(\mathbf S_{\hat{\mathcal{P}}, \hat{\mathcal{P}}}\left(\mathbf x\right)\right)^{\frac{1}{2}} \left[\hat{\mathbf{A}}^{\hat{\mathcal{P}}} \left(\mathbf x\right)\lambda_1 \left(\hat{\mathbf{\Gamma}}^{\hat{\mathcal{P}}}\left(\mathbf x\right)\right) ^\top\right] \left(\mathbf S_{\hat{\mathcal{P}}^c, \hat{\mathcal{P}}^c}\left(\mathbf x\right)\right)^{\frac{1}{2}}\\ &=\lambda_1 \sqrt{s_{\hat{\mathcal{P}}^c, \hat{\mathcal{P}}^c}\left(\bx\right)} \left(\bS_{{\hat{\mathcal{P}}}, {\hat{\mathcal{P}}}}\left(\bx\right)\right)^{\frac{1}{2}}\frac{1}{\left\| \whitens_{\hat{\mathcal{P}}, \hat{\mathcal{P}}^c} \left(\bx\right)\right\|_2}\whitens_{\hat{\mathcal{P}}, \hat{\mathcal{P}}^c} \left(\bx\right)\\
              &= \frac{\lambda_1}{\hat\lambda_1 ^{\hat{\mathcal P}}\left(\mathbf x\right)} \bs_{\hat{\mathcal{P}}, \hat{\mathcal{P}}^c}\left( \bx\right),
    \end{aligned}
\end{equation}
where the last equality follows from \eqref{eqn:SVD_r2}. Next, we look at the off-diagonal blocks of the correlation matrices of $\mathbf S'\left(\mathbf x, \hat{\mathcal{P}}, \left\{\lambda_1\right\}\right)$ and $\bS \left( \bx\right)$. We recall that $\mathbf R \left( \mathbf S\right)$ denotes the correlation matrix corresponding to the covariance matrix $\mathbf S$, as defined in Section \ref{subsection:thresholding_procedure}. For $i' \in \hat{\mathcal P}$ and $j' \in \hat{\mathcal P}^c$,
\begin{equation}
    \label{eqn:rp=1_correlation}
    \begin{aligned}
          \bR_{i',j'}\left(\mathbf S'\left(\mathbf x, \hat{\mathcal{P}}, \left\{\lambda_1\right\}\right)\right) &=  \frac{\mathbf S_{i',j'}'\left(\mathbf x, \hat{\mathcal{P}}, \left\{\lambda_1\right\}\right)}{\sqrt{\mathbf S_{i',i'}'\left(\mathbf x, \hat{\mathcal{P}}, \left\{\lambda_1\right\}\right)\mathbf S_{j',j'}'\left(\mathbf x, \hat{\mathcal{P}}, \left\{\lambda_1\right\}\right)}}\\
&= \frac{\lambda_1}{\hat\lambda_1 ^{\hat{\mathcal P}}\left(\mathbf x\right)} \frac{\mathbf S_{i',j'}\left(\mathbf x\right)}{\sqrt{\mathbf S_{i',i'}\left(\mathbf x\right)\mathbf S_{j',j'}\left(\mathbf x\right)}}\\
&= \frac{\lambda_1}{\hat\lambda_1 ^{\hat{\mathcal P}}\left(\mathbf x\right)} \bR_{i',j'}\left(\mathbf S \left( \mathbf x \right)\right),
\end{aligned}
\end{equation}
where the second inequality follows from combining \eqref{eqn:rp=1} and the fact that the diagonal elements of $\mathbf S'\left(\mathbf x, \hat{\mathcal{P}}, \left\{\lambda_1\right\}\right)$ are equal to those of $\mathbf S\left(\mathbf x\right)$. Next we focus on characterizing the conditioning set defined in \eqref{eqn:polytope_proof}.
\begin{equation}
\label{eqn:condition_rp1_simple}
\begin{aligned}
 {\mathcal{G}}\left(\mathbf x; \hat{\mathcal{P}} \right) 
 &= \left\{ \lambda_1 \in [0, 1] :  \:c \geqslant \max_{i' \in \hat{\mathcal P}, j' \in \hat{\mathcal P}^c}\left|\bR_{i',j'}\left(\mathbf S'\left(\mathbf x, \hat{\mathcal{P}}, \left\{\lambda_1\right\}\right)\right)\right|\right\}\\
 &=\left\{ \lambda_1 \in [0, 1] :  \:c \geqslant \frac{\lambda_1}{\hat\lambda_1 ^{\hat{\mathcal P}}\left(\mathbf x\right)} \max_{i' \in \hat{\mathcal P}, j' \in \hat{\mathcal P}^c}\left| \bR_{i',j'}\left(\mathbf S\left( \bx\right)\right)\right|\right\}\\
 &=\left\{ \lambda_1 \in [0, 1] :  \: \lambda_1 \leqslant \frac{c\hat\lambda_1 ^{\hat{\mathcal P}}\left(\mathbf x\right)}{\max_{i' \in \hat{\mathcal P}, j' \in \hat{\mathcal P}^c}\left| \bR_{i',j'}\left(\mathbf S\left( \bx\right)\right)\right|}\right\}\\
 &= \left\{ \lambda_1 \geqslant 0:  \lambda_1^2 \leqslant \min \left\{1, \left( \frac{c\hat\lambda_1 ^{\hat{\mathcal P}}\left(\mathbf x\right)}{\max_{i' \in \hat{\mathcal P}, j' \in \hat{\mathcal P}^c}\left| \bR_{i',j'}\left(\mathbf S\left( \bx\right)\right)\right|}\right)^2\right\}\right\}\\
 &= \left\{ \lambda_1 \geqslant 0:  \lambda_1^2 \leqslant g_u\left(\mathbf S \left( \mathbf x\right), \hat{\mathcal{P}}, c\right)\right\},
\end{aligned}
\end{equation}
where $g_u\left(\mathbf S \left( \mathbf x\right), \hat{\mathcal{P}}, c\right) := \min \left\{1, \left( \frac{c\hat\lambda_1 ^{\hat{\mathcal P}}\left(\mathbf x\right)}{\max_{i' \in \hat{\mathcal P}, j' \in \hat{\mathcal P}^c}\left| \bR_{i',j'}\left(\mathbf S\left( \bx\right)\right)\right|}\right)^2\right\}$.

We can write the denominator of the first equality of \eqref{eqn:probability} as
$$
\begin{aligned}
          \int_{\lambda_1 \in {\mathcal{G}}\left(\mathbf x; \hat{\mathcal{P}} \right)} f\left(\lambda_1\right)d\lambda_1 &= \int_{\left\{ \lambda_1 \geqslant 0\right\} \cap \left\{  \lambda_1^2 \leqslant g_u\right\}} f\left(\lambda_1\right)d\lambda_1\\ &= \int_{ \left\{  \lambda_1^2 \leqslant g_u\right\}} f\left(\lambda_1\right)d\lambda_1\\
          &= \int_{0}^{g_u} f_{\frac{(p - 1)}{2}, \frac{(n - p)}{2}}(t)dt,
\end{aligned}
$$

\noindent where $f_{\alpha, \beta}(\cdot)$ is the probability density function of the $\mathrm{Beta}\left(\alpha, \beta\right)$ distribution. Here, the first equality follows from \eqref{eqn:condition_rp1_simple}, the second equality follows from the fact that $f(\lambda_1) = 0$ if $\lambda_1 < 0$, and the third equality follows from a change of variable and (iii) of Proposition \ref{prop:density_singular_values}. 

Next, we concentrate on the numerator in \eqref{eqn:probability}. For $r\left(\hat{\mathcal P}\right) = 1$, \\$\left\{\prod_{i= 1}^{r\left({\hat{\mathcal P}}\right)}  \left(1 - \lambda_i^2\right) \leqslant  \prod_{i = 1}^{r\left({\hat{\mathcal P}}\right)} \left(1 - \hat\lambda_i^{\hat{\mathcal P}}\left(\mathbf x\right)^2\right)\right\}$ is equivalent to $\lambda_1^2 \geqslant \hat\lambda_1^{\hat{\mathcal P}}\left(\mathbf x\right)^2$. Proceeding as in the case of the denominator, $p\left(\mathbf x; \hat{\mathcal{P}}\right) $ in \eqref{eqn:probability} can be rewritten as

$$
p\left(\mathbf x; \hat{\mathcal{P}}\right) = \frac{\int_{\min\left\{g_u, \hat\lambda_1^{\hat{\mathcal P}}\left(\mathbf x\right)^2\right\}}^{g_u} f_{\frac{(p - 1)}{2}, \frac{(n - p)}{2}}(t)dt} {\int_{0}^{g_u} f_{\frac{(p - 1)}{2}, \frac{(n - p)}{2}}(t)dt},
$$
where $f_{\alpha, \beta}(\cdot)$ is the probability density function of the $\mathrm{Beta}\left(\alpha, \beta\right)$ distribution.

\subsection{Numerical integration to approximate \eqref{eqn:probability} when $r\left(\hat{\mathcal P}\right)$ is small}
\label{Appendix:geometry_integration}
For $r\left(\hat{\mathcal P}\right) > 1$, we opt for numerical integration. Given $\mathbf L$ and $\mathbf g$ (defined in Theorem \ref{thm:character}), we evaluate $p\left(\mathbf x; \hat{\mathcal{P}}\right)$ in \eqref{eqn:probability} through the following steps:

\begin{enumerate}
    \item We compute the convex hull representation of the convex polytope given by $\left\{\bm\lambda : \mathbf L \bm\lambda \leqslant \mathbf g \right\}$. This is achieved with the function \texttt{scdd}$\left(\cdot\right)$ in the R package \texttt{rcdd} \citep{rcdd}.
    
    \item Next, we find the Delaunay triangulation of the convex hull. We use the function \texttt{delaunayn}$\left(\cdot\right)$ in the R package \texttt{geometry} \citep{geometry} to do this. 
    
    \item Let $\mathcal{T}_1, \mathcal{T}_2, \ldots, \mathcal{T}_N$ be the simplices of the triangulation. We integrate the functions $f\left(\lambda_1, \lambda_2, \ldots, \lambda_{r\left(\hat{\mathcal{P}}\right)}\right)\mathds{1}\left\{\prod_{i= 1}^{r\left({\hat{\mathcal P}}\right)}  \left(1 - \lambda_i^2\right) \leqslant  \prod_{i = 1}^{r\left({\hat{\mathcal P}}\right)} \left(1 - \hat\lambda_i^{\hat{\mathcal P}}\left(\mathbf x\right)^2\right)\right\}$ (corresponding to the numerator in \eqref{eqn:probability}) and $f\left(\lambda_1, \lambda_2, \ldots, \lambda_{r\left(\hat{\mathcal{P}}\right)}\right)$ (corresponding to the denominator in \eqref{eqn:probability}) on the obtained simplices and take their corresponding sums. We use the function \texttt{adaptIntegrateSimplex}$\left(\cdot\right)$ in the R package \texttt{SimplicialCubature} \citep{SimplicialCubature} for this purpose.
    \item Finally we compute $p\left(\mathbf x; \hat{\mathcal{P}}\right)$ as the ratio of these two approximations,

\begin{equation}
\label{eqn:Delaunay_integration}
\begin{aligned}
    &p\left(\mathbf x; \hat{\mathcal{P}}\right) \\
    &= \frac{\sum_{t = 1}^N \int_{\mathcal T_t} f\left(\lambda_1, \ldots, \lambda_{r\left(\hat{\mathcal{P}}\right)}\right)\mathds{1}\left\{\prod_{i= 1}^{r\left({\hat{\mathcal P}}\right)}  \left(1 - \lambda_i^2\right) \leqslant  \prod_{i = 1}^{r\left({\hat{\mathcal P}}\right)} \left(1 - \hat\lambda_i^{\hat{\mathcal P}}\left(\mathbf x\right)^2\right)\right\}d\lambda_1\ldots d\lambda_{r\left(\hat{\mathcal{P}}\right)}}{\sum_{t = 1}^N \int_{\mathcal T_t} f\left(\lambda_1, \ldots, \lambda_{r\left(\hat{\mathcal{P}}\right)}\right)d\lambda_1\ldots d\lambda_{r\left(\hat{\mathcal{P}}\right)}}.
\end{aligned}
\end{equation}
\end{enumerate}

\subsection{Monte Carlo approximation of \eqref{eqn:probability} when $r\left(\hat{\mathcal P}\right)$ is large}
\label{Appendix:algorithm}
We use  $B = 1,000$ for the Monte Carlo approximation in Section \ref{sec:mainrp_MC}, unless $\sum_{b = 1}^{1000} \mathds{1}\left\{\mathbf{L}{\bm\lambda^{\left(b\right)}} \leqslant \mathbf{g}\right\} < 100$, in which case we use a larger value of $B$. 

\begin{algorithm}[H]
	\caption{Computation of $p\left(\mathbf x; \hat{\mathcal{P}}\right) $ in \eqref{eqn:thm_main}}\label{algo:SI}
	 \textbf{Input:} Sample covariance matrix $\bS (\bx)$; a group of variables $\hat{\mathcal{P}} \in \mathcal{P}(\bS(\bx))$, where $\mathcal{P}(\bS(\bx))$ is computed as described in Algorithm \ref{algo:Selection_procedure}, using a threshold $c$; the sample size $n$; the number of Monte Carlo replicates, $B$.
	\begin{algorithmic}[1]
		\Procedure{}{}
		\State Compute $\hat{\bm{\Lambda}}^{\hat{\mathcal P}}(\bx)$ from (\ref{eqn:CCA}) with diagonal elements $\{ \hat\lambda_i^{\hat{\mathcal P}} (\bx)\}_{i = 1}^{r(\hat{\mathcal P})}$;
		\If{$r \left( \hat{\mathcal{P}}\right) = 1$}
		\State Compute $g_u$ as defined in Proposition \ref{prop:univariate_beta};
		\State Compute $p(\bx; \hat{\mathcal{P}})$ from \eqref{eqn:univariate_beta_proposition};
		\EndIf
		\If{$ 1 < r \left( \hat{\mathcal{P}}\right) \leqslant 5$}
		\State Compute $\mathbf{L}(\mathbf S (\mathbf x), \hat{\mathcal{P}})$ and $\mathbf{g}(\hat{\mathcal{P}}, c)$ as defined in Theorem \ref{thm:character};
		\State Compute the convex hull representation of the polytope $\left\{\bm\lambda : \mathbf L \bm\lambda \leqslant \mathbf g \right\}$ ;
		\State Compute Delaunay triangulation of the convex hull;
		\State Identify the simplices of the triangulation,  $\mathcal T_1, \mathcal T_2, \ldots, \mathcal T_N$;
		\For{$t = 1, 2, \ldots, N$}
		\State Compute $\int_{T_t} f\left(\lambda_1, \ldots, \lambda_{r\left(\hat{\mathcal{P}}\right)}\right)d\lambda_1\ldots d\lambda_{r\left(\hat{\mathcal{P}}\right)}$ and \newline  \hspace*{4em} $\int_{T_t} f\left(\lambda_1, \ldots, \lambda_{r\left(\hat{\mathcal{P}}\right)}\right)\mathds{1}\left\{\prod_{i= 1}^{r\left({\hat{\mathcal P}}\right)}  \left(1 - \lambda_i^2\right) \leqslant  \prod_{i = 1}^{r\left({\hat{\mathcal P}}\right)} \left(1 - \left(\hat\lambda_i^{\hat{\mathcal P}}\left(\mathbf x\right)\right)^2\right)\right\}d\lambda_1\ldots d\lambda_{r\left(\hat{\mathcal{P}}\right)}$ \newline  \hspace*{4em} with $f\left( \cdot\right)$ is given by \eqref{eqn:density_singular_values} ;
		\EndFor
		\State Compute $p(\bx; \hat{\mathcal{P}})$ from \eqref{eqn:Delaunay_integration};
		\EndIf
		\If{$ r \left( \hat{\mathcal{P}}\right) > 5$}
		\State Compute $\mathbf{L}(\mathbf S (\mathbf x), \hat{\mathcal{P}})$ and $\mathbf{g}(\hat{\mathcal{P}}, c)$ as defined in Theorem \ref{thm:character};
		\For{$b = 1, 2, \ldots, B$}
		\State Simulate $\bm\lambda^{(b)}$ from (ii) of Proposition \ref{prop:density_singular_values};
		\EndFor
		\State Approximate $p(\bx; \hat{\mathcal{P}})$ with $\hat p_\mathrm{MC}(\bx; \hat{\mathcal{P}})$ from (\ref{eqn:Naive_MC}).
		\EndIf
		\EndProcedure
	\end{algorithmic}
\end{algorithm}

Furthermore, due to numerical instability, the numerical integration approach described in Section \ref{sec:mainrp_small} on rare occasions may produce estimates of $p\left(\mathbf x; \hat{\mathcal{P}}\right)$ that are outside of $[0, 1]$. In those cases, we use the Monte Carlo approach from Section \ref{sec:mainrp_MC}.

\section{Simulation setup}
\label{Appendix:Simulation}
\subsection{Generation of $\bSigma$}
We generate $\bSigma$ as $\bSigma = \bSigma_0 + w\mathbf 1 \mathbf 1^\top$, where $\bSigma_0$ is a $p \times p $ positive definite matrix, randomly generated  using the function \texttt{genPositiveDefMat}$\left(\cdot\right)$ with default options in the \texttt{R} package \texttt{clusterGeneration} (\cite{clusterGenerationR}), and $w$ is a $\mathrm{Unif}[0, 1]$ random variable. This code first generates uniformly distributed eigenvalues on $[1, 10]$, and then uses columns of a randomly generated orthogonal matrix as eigenvectors to construct the covariance matrix $\bSigma_0$. Since this produces $\bSigma_0$ with small off-diagonal elements (i.e. low signal strength), we increase the signal strength by adding $w\mathbf 1 \mathbf 1^\top$ to $\bSigma_0$. 
\subsection{Choice of threshold $c$}
For each $\bSigma$, the threshold $c$ is adaptively chosen based on $\bSigma$. First, let $\tilde c$ be the threshold that produces exactly two connected components when we apply the thresholding procedure described in Section  \ref{subsection:thresholding_procedure} on the population correlation matrix $\mathrm{diag}\left(\bSigma\right)^{- \frac{1}{2}} \bSigma \mathrm{diag}\left(\bSigma\right)^{- \frac{1}{2}}$. To account for the deviation of the sample correlation matrix from the population correlation matrix, we set the threshold $c\left(\bSigma\right) := \max\left\{\tilde c + \sqrt{\log\left(p\right)/n}, 1-10^{-8}\right\}$. 

\end{document}